\PassOptionsToPackage{unicode}{hyperref}
\PassOptionsToPackage{hyphens}{url}
\PassOptionsToPackage{dvipsnames,svgnames,x11names}{xcolor}
\documentclass[
  11pt]{article}

 \usepackage{bm, commath}
\usepackage{natbib}
\usepackage{caption}
\usepackage{graphicx}
\usepackage{subcaption}
\usepackage{amsmath, amsfonts, amsthm,amssymb}
\usepackage{float}
\usepackage{booktabs,siunitx}
\usepackage{url}
\usepackage{multirow}
\usepackage{amssymb}
\usepackage{bm}
\usepackage{mathrsfs}
\usepackage{enumitem}
\usepackage{threeparttable}
\usepackage{soul}
\allowdisplaybreaks

\usepackage{listings}
\usepackage{bbm}
\usepackage{textcomp}
\usepackage[ruled]{algorithm2e}
\SetKwInput{KwParam}{Parameter}
\SetAlgoCaptionLayout{centerline}

\usepackage[utf8]{inputenc}
\usepackage{float}
\usepackage{tikz}
\usepackage{bbding}
\usetikzlibrary{trees}
\usetikzlibrary{shapes,decorations,arrows,calc,arrows.meta,fit,positioning}
\usetikzlibrary{shapes.multipart}
\tikzset{
    -Latex,auto,node distance =1 cm and 1 cm,semithick,
    state/.style ={ellipse, draw, minimum width = 0.7 cm},
    state1/.style ={ draw, minimum width = 0.7 cm},
    point/.style = {circle, draw, inner sep=0.04cm,fill,node contents={}},
    bidirected/.style={Latex-Latex,dashed},
    el/.style = {inner sep=2pt, align=left, sloped}
}

\newcommand{\indep}{\rotatebox[origin=c]{90}{$\models$}}
\newtheorem{theorem}{Theorem}

\newtheorem{assumption}{Assumption}

\newtheorem*{assumptionone*}{Assumption 1$^*$}

\usepackage{mathtools}

\usepackage{xcolor}

\makeatletter
\renewcommand{\algocf@captiontext}[2]{#1\algocf@typo. \AlCapFnt{}#2} 

\def\bmX{\bm{X}}
\def\bmx{\bm{x}}

\makeatother

\usepackage{longtable,booktabs,array}
\usepackage{calc} 
\usepackage{etoolbox}
\makeatletter
\patchcmd\longtable{\par}{\if@noskipsec\mbox{}\fi\par}{}{}
\makeatother
\IfFileExists{footnotehyper.sty}{\usepackage{footnotehyper}}{\usepackage{footnote}}
\makesavenoteenv{longtable}
\usepackage{graphicx}
\makeatletter
\def\maxwidth{\ifdim\Gin@nat@width>\linewidth\linewidth\else\Gin@nat@width\fi}
\def\maxheight{\ifdim\Gin@nat@height>\textheight\textheight\else\Gin@nat@height\fi}
\makeatother
\setkeys{Gin}{width=\maxwidth,height=\maxheight,keepaspectratio}
\makeatletter
\def\fps@figure{htbp}
\makeatother

\makeatletter
\@ifpackageloaded{caption}{}{\usepackage{caption}}
\AtBeginDocument{%
\ifdefined\contentsname
  \renewcommand*\contentsname{Table of contents}
\else
  \newcommand\contentsname{Table of contents}
\fi
\ifdefined\listfigurename
  \renewcommand*\listfigurename{List of Figures}
\else
  \newcommand\listfigurename{List of Figures}
\fi
\ifdefined\listtablename
  \renewcommand*\listtablename{List of Tables}
\else
  \newcommand\listtablename{List of Tables}
\fi
\ifdefined\figurename
  \renewcommand*\figurename{Figure}
\else
  \newcommand\figurename{Figure}
\fi
\ifdefined\tablename
  \renewcommand*\tablename{Table}
\else
  \newcommand\tablename{Table}
\fi
}
\@ifpackageloaded{float}{}{\usepackage{float}}
\floatstyle{ruled}
\@ifundefined{c@chapter}{\newfloat{codelisting}{h}{lop}}{\newfloat{codelisting}{h}{lop}[chapter]}
\floatname{codelisting}{Listing}

\makeatother
\makeatletter
\@ifpackageloaded{caption}{}{\usepackage{caption}}
\@ifpackageloaded{subcaption}{}{\usepackage{subcaption}}
\makeatother

\ifLuaTeX
  \usepackage{selnolig}  
\fi
\usepackage[]{natbib}
\usepackage{bookmark}

\IfFileExists{xurl.sty}{\usepackage{xurl}}{} 
\hypersetup{
  pdftitle={Title},
  pdfauthor={Author 1; Author 2},
  pdfkeywords={3 to 6 keywords, that do not appear in the title},
  colorlinks=true,
  linkcolor={blue},
  filecolor={Maroon},
  citecolor={Blue},
  urlcolor={Blue},
  pdfcreator={LaTeX via pandoc}}

\newcommand{\anon}{1}

\usepackage{xr}

\begin{document}

\def\spacingset#1{\renewcommand{\baselinestretch}%
{#1}\small\normalsize} \spacingset{1}


\if1\anon
{
  \title{\bf Proximal Learning for Trials With External Controls: A Case Study in HIV Prevention}
  \author{Yilin Song$^1$, Yinxiang Wu$^2$, Raphael J. Landovitz$^3$, Susan Buchbinder$^4$, \\
  Srilatha Edupuganti$^5$, Lydia Soto-Torres$^6$, Kendrick Li$^7$, Xu Shi$^8$, Fei Gao$^9$, \\
  Deborah Donnell$^9$,
  Holly Janes$^8$, and Ting Ye$^2$\thanks{Correspond to tingye1@uw.edu.}\vspace{.4cm}\\ 
    $^1$Department of Biostatistics, Columbia University \\
    $^2$Department of Biostatistics, University of Washington \\
    $^3$David Geffen School of Medicine, University of California, Los Angeles\\
    $^4$Bridge HIV, San Francisco Department of Public Health\\
    $^5$Department of Medicine, Emory University\\
    $^6$Division of AIDS, National Institute of Allergy and Infectious Diseases\\
    $^7$St. Jude Children's Research Hospital\\
    $^8$Department of Biostatistics, University of Michigan\\
    $^9$Vaccine and Infectious Disease Division, \\Fred Hutchinson Cancer Research Center
    }
    \date{}
  \maketitle
} \fi

\if0\anon
{
  \bigskip
  \bigskip
  \bigskip
  \begin{center}
    {\LARGE \bf Proximal Learning for Trials With External Controls: A Case Study in HIV Prevention}
\end{center}
  \medskip
} \fi

\bigskip
\vspace{-1.2cm}
\begin{abstract}
With the advent of effective pre-exposure prophylaxis agents, active-controlled HIV prevention trials have become a common study design. Nevertheless, estimating absolute efficacy relative to a placebo remains important. In this paper, we introduce a novel application of proximal causal inference methods to estimate the counterfactual cumulative HIV incidence under placebo for participants in an active-controlled trial of cabotegravir, using external control data from a placebo-controlled trial with similar eligibility criteria. We leverage baseline sexually transmitted infection status and geographic region as negative control outcome and exposure variables, respectively. We address two key challenges: unmeasured differences in HIV risk between trials and statistical difficulties arising from low HIV incidence rates in both studies.  To overcome these challenges, we develop two proximal inference approaches: (1) a semiparametric inverse probability of censoring weighting estimator, and (2) a two-stage regression-based strategy tailored to low-event-rate settings. Our theoretical and numerical investigations demonstrate these methods yield reliable estimates of the counterfactual one-year cumulative HIV incidence under placebo, and provide robust evidence of the superior efficacy of cabotegravir compared with placebo. These findings highlight the potential of proximal inference methods to estimate placebo-controlled effects in both single-arm and active-controlled trials by leveraging external controls. 
\end{abstract}

\vspace{-0.2cm}
\noindent%
{\it Keywords:} active-controlled trials, causal inference, censoring, data integration
\vfill

\newpage
\spacingset{1.8} 

\section{HIV Prevention Trials}
\vspace{-0.5cm}\subsection{HPTN 083 as the primary trial of interest}\vspace{-0.3cm}

HPTN 083 (ClinicalTrials.gov: NCT02720094) is a Phase 2b/3 randomized clinical trial evaluating the safety and efficacy of long-acting injectable cabotegravir compared to daily oral tenofovir disoproxil fumarate/emtricitabine (TDF/FTC),  the first biomedical agent proven effective for HIV prevention through pre-exposure prophylaxis (PrEP) \citep{HPTN083}.
The trial enrolled participants from December 2016 through March 2020 across the United States, Latin America, Asia, and Africa. The intention-to-treat (ITT) population included 4,566 cisgender men and transgender women who have sex with men, with 2,282 participants randomized to the cabotegravir arm and 2,284 to the TDF/FTC arm.

In the ITT population, five participants (two in the cabotegravir arm and three in the TDF/FTC arm) were retrospectively found to have HIV infection at baseline, and 71 participants had no follow-up visits after enrollment. All of these participants were excluded from the primary efficacy analysis. The trial demonstrated the superiority of cabotegravir over TDF/FTC in preventing HIV infection. The one-year cumulative HIV incidence was 0.41\% (95\% CI: 0.20 to 0.70) in the cabotegravir group compared to 1.22\% (95\% CI: 0.90 to 1.70) in the TDF/FTC group.

\vspace{-0.5cm}\subsection{Estimating counterfactual HIV incidence under placebo}\vspace{-0.3cm}

Active-controlled HIV prevention trials, such as HPTN 083, have become a common design following the availability of effective PrEP agents like TDF/FTC \citep{DISCOVER}. Given the well-established efficacy of these agents, conducting placebo-controlled trials that withhold proven interventions requires careful ethical justification \citep{unaids-ethics}. Furthermore, the rapid evolution of HIV prevention strategies presents additional challenges. Participants assigned to a placebo arm {must be} offered prevention alternatives within the trial, and many may also seek other prevention options outside the trial, making it difficult to isolate the effect of the investigational agent.

Despite the absence of a concurrent placebo group, estimating the absolute efficacy of a new intervention remains highly relevant. {This involves comparing the observed outcome under a new intervention to the outcome expected under a clearly-defined placebo condition -- the \emph{counterfactual placebo incidence} \citep{glidden2019advancing, glidden2020bayesian, temple2000placebo}.} There are two primary motivations for such estimation. First, it enables patients and clinicians to make informed decisions by providing a clear understanding of the true level of protection offered by different interventions. Second, as PrEP agents become more effective and HIV incidence continues to decline, conducting active-controlled trials with HIV incidence as the primary endpoint becomes increasingly impractical and/or infeasible.

{In the context of HIV prevention, several study design strategies have been proposed to estimate the absolute efficacy of an intervention in the absence of a placebo arm. \cite{donnell2024studydesign} described six potential study design strategies as well as their strengths and limitations when using data from: (1) registrational cohorts, (2) recency assays, (3) {historical} external placebo arms, (4) HIV incidence biomarkers, (5) drug concentrations, and (6) immune biomarkers. Several of these strategies have been implemented in actual trials. For example, the PrEPVacc study estimates counterfactual placebo incidence using a pre-trial registration cohort \citep{PrEPVaccNCT04066881, avertedinfectionsratio, PrEPVaccKansiime2025}. The recency assay approach estimates baseline HIV incidence based on the proportion of recent infections among individuals screened for trial participation in cross-sectional samples; this method has been applied in Gilead's PURPOSE1 and PURPOSE2 trials \citep{PURPOSE1, PURPOSE2}. \cite{Donnell2023} employed the external placebo arm approach, using direct standardization to adjust for observed covariates. Considerations of the external placebo arms approach have also been discussed in the draft guidance from the US Food and Drug Administration (FDA) \citep{fda2023externallycontrolled}. 
The HIV incidence biomarkers approach, exemplified by \cite{Mullick2020} and \cite{Zhu2024counterfactual}, involves fitting a model {on historical data linking study-level gonorrhea and HIV incidence rates and to predict counterfactual placebo HIV incidence in the current primary study. This approach is complicated by the power of antiretroviral therapy as prevention in a population, potentially unlinking sexually transmitted infection (STI) transmission from HIV transmission, and is now infeasible for some populations because of the demonstrated efficacy of doxycycline post-exposure prophylaxis.} Researchers have also explored other active-controlled trial designs \citep{hollydesignreview, gao2025design}. Together, these design strategies provide an expanding array of options for estimating counterfactual placebo incidence and designing HIV prevention trials when including a placebo arm is not feasible. }



In this study, we focus on estimating the counterfactual 
cumulative HIV incidence under placebo for participants in the HPTN 083 trial, drawing on external control data from a placebo-controlled study. By comparing these counterfactual estimates with the observed cumulative HIV incidence in the cabotegravir arm of HPTN 083, we obtain estimates of absolute efficacy. However, this approach presents methodological challenges that require further statistical development.

\vspace{-0.5cm}\subsection{HVTN704/HPTN 085 as an external control dataset}\vspace{-0.3cm}\label{amp_external_intro}

The Antibody Mediated Prevention (AMP) study (also known as HVTN 704/HPTN 085, ClinicalTrials.gov: NCT02716675) was conducted in the United States, Peru, Brazil, Switzerland from April 2016 through October 2018. The AMP study had similar eligibility criteria to HPTN 083 and enrolled at-risk cisgender men and transgender individuals to evaluate the safety, tolerability, and efficacy of the passively-infused VRC01 broadly neutralizing antibody (bnAb) for HIV prevention.

The AMP study included three arms: (1) high-dose arm receiving VRC01 at 30 mg/kg, (2) low-dose arm receiving VRC01 at 10 mg/kg, and (3) placebo arm with sterile saline infusions. HIV cumulative incidence through 1 year (\%) was 2.5 (95\% CI: [1.71, 3.53]), 2.2 (95\% CI: [1.46, 3.18]), and 2.98 (95\% CI: [2.11, 4.09]) in the low-dose, high-dose, and placebo groups, respectively.  We used the placebo arm as our negative control data and excluded $n=11$ participants who had no follow-up visits after randomization. Details of the data processing steps are provided in Figure \ref{fig:flow}. 

\begin{figure}[!htbp]
    \centering
    \includegraphics[width=\linewidth]{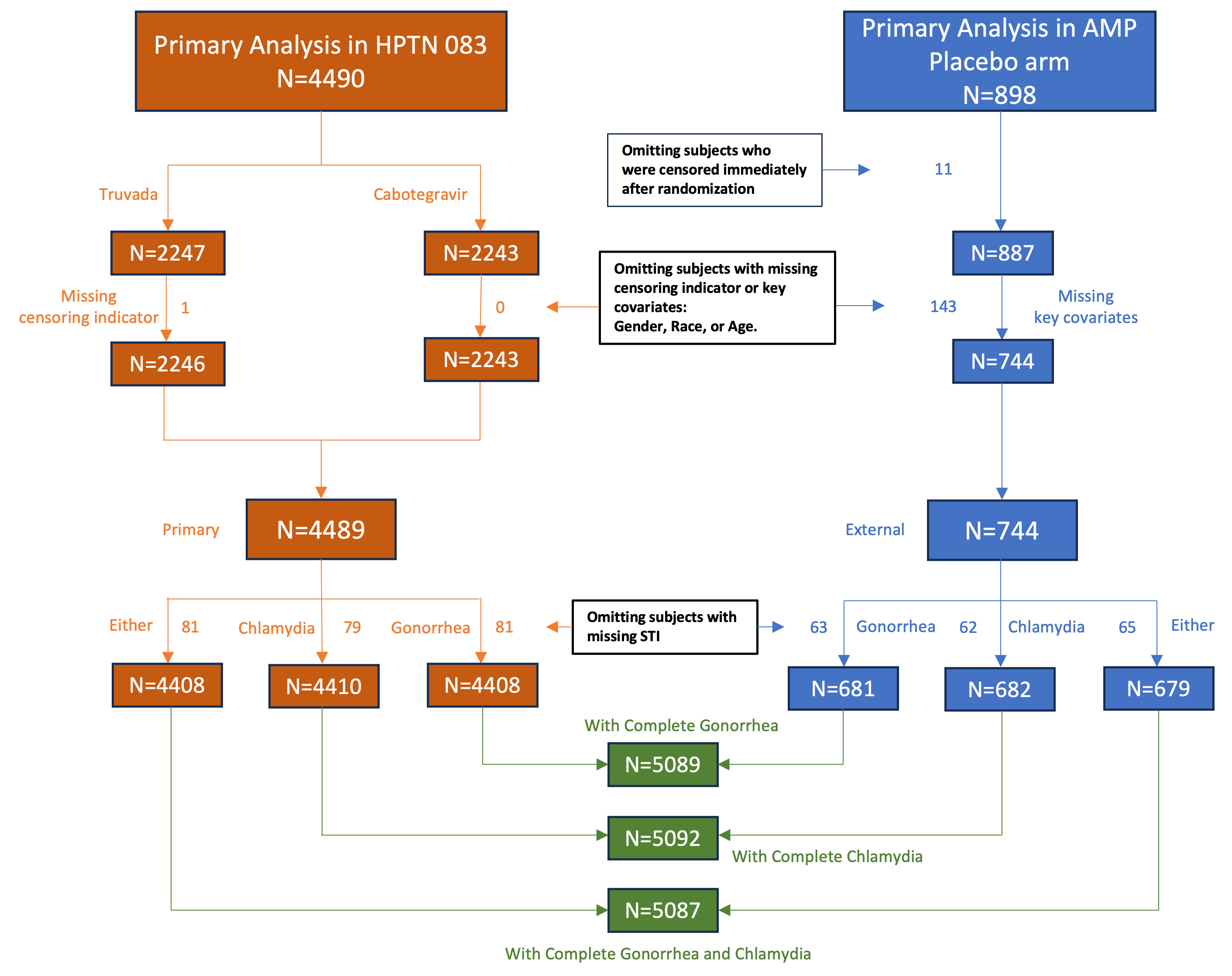}
    \caption{Data processing workflow. We first excluded observations with missing key covariates, including gender, race, and age. We then constructed three separate analytic datasets by excluding individuals with missing baseline gonorrhea, missing baseline chlamydia, or missing values for either gonorrhea or chlamydia, respectively. Baseline STI measures were used as negative control outcomes.}
    \label{fig:flow}
\end{figure}

It is important to clearly define the placebo condition in the AMP study in order to properly interpret the counterfactual placebo incidence we are estimating for HPTN 083. In AMP, follow-up data showed that 39.0\% (95\% CI: [35.9, 42.0]) of person-years had {detectable\footnote{Defined as detection at or above the lower limit of quantification}} TDF/FTC and 28.9\% (95\% CI: [26.4, 31.6]) had {effective\footnote{Defined as detection of at least 700 fmol/punch}} TDF/FTC, with these rates evenly distributed across treatment groups. This illustrates the reality in modern HIV prevention trials, where few placebo-controlled studies are conducted without concomitant use of biomedical prevention. {In comparison, in the HPTN 083 TDF/FTC arm, follow-up data showed the TDF/FTC concentrations consistent
with receipt of at least four TDF-FTC
doses per week were detected in 72.3\% of the samples.} Therefore, the placebo condition in AMP reflects the use of additional biomedical prevention, which is arguably a clinically relevant definition of placebo in this setting.

To assess the comparability of the two trials, we examined baseline characteristics of participants in HPTN 083 and in the placebo arm of AMP (Table \ref{table:table1}). Because variables measured in different trials are rarely identical, even when conducted within the same network and during similar time periods, we harmonized variable definitions through discussions with the protocol teams of each trial. { As shown in Table \ref{table:table1}, the placebo arm of AMP included a higher proportion of White participants and a greater percentage of individuals above 30 years of age. In contrast, HPTN 083 enrolled more Black participants, and {more cisgender men}. The prevalence of baseline STIs, including rectal gonorrhea and chlamydia, was also higher in HPTN 083. These variables are important to adjust for in the analysis because they may be associated with the risk of HIV infection \citep{race-HIV, STI-HIV}.  }

This comparison highlights a central challenge in using external controls. Systematic differences between trial populations can introduce confounding bias, and even with rigorous adjustment for observed covariates, residual confounding from unmeasured variables may persist. For example, {regional HIV prevalence,} viral load among partners, and local sexual network density are all important contextual factors influencing individual risk of infection, yet they are difficult to measure and were not recorded in either study. Because the two trials recruited participants from different geographic regions, unmeasured differences in local HIV transmission dynamics may contribute to between-study differences in HIV incidence under placebo that are not fully explained by measured variables.

\begin{table}[!htbp]
\centering
\caption{Baseline characteristics of participants in HPTN 083 and the AMP placebo arm. Demographics were collected via baseline questionnaires. Baseline gonorrhea and chlamydia infections were determined through laboratory testing of rectal swabs.
}
\label{table:table1}        
\resizebox{0.8\textwidth}{!}{
\begin{tabular}{lcc}
\toprule
\vspace{-0.2em}
 & HPTN 083 ($n = 4490$) & AMP Placebo ($n = 898$) \\
\midrule
Age, $n$ (\%) & & \vspace{-0.7em}\\
\hspace{1em} 18-20 years & 634 (14.1) & 95 (10.6) \vspace{-0.8em}\\
\hspace{1em} 21-30 years & 2557 (56.9) & 493 (54.9) \vspace{-0.8em}\\
\hspace{1em} 31-40 years & 907 (20.2) & 230 (25.6) \vspace{-0.8em}\\
\hspace{1em} $\geq$ 40 years & 392 (8.7) & 80 (8.9) \vspace{-0.4em}\\ 

Race, $n$ (\%) & & \vspace{-0.7em}\\
\hspace{1em} White & 1231 (27.4) & 285 (31.7) \vspace{-0.8em}\\
\hspace{1em} Black & 1074 (23.9) & 131 (14.6) \vspace{-0.8em}\\
\hspace{1em} Other & 2185 (48.7) & 482 (53.7) \vspace{-0.4em}\\

Gender, $n$ (\%) & & \vspace{-0.7em}\\
\hspace{1em} Female or Trans Female & 480 (10.7) & 33 (3.7) \vspace{-0.8em}\\
\hspace{1em} Male & 3899 (86.8) & 688 (76.6) \vspace{-0.8em}\\
\hspace{1em} Other & 111 (2.5) & 29 (3.2) \vspace{-0.8em}\\
\hspace{1em} Missing & 0 (0.0) & 148 (16.5) \vspace{-0.4em}\\

Region, $n$ (\%) & & \vspace{-0.7em}\\
\hspace{1em} Non-Latin America & 2562 (57.1) & 470 (52.3)\vspace{-0.8em}\\
\hspace{1em} Latin America & 1928 (42.9) & 428 (47.6) \vspace{-0.4em}\\

Ethnicity, $n$ (\%) & & \vspace{-0.7em}\\
\hspace{1em} Hispanic/Latino(a) & 2065 (46.0) & 529 (58.9) \vspace{-0.8em}\\
\hspace{1em} Non-Hispanic/Latino(a) & 2424 (54.0) & 369 (41.1) \vspace{-0.8em}\\
\hspace{1em} Missing & 1 (0.0) & 0 (0.0) \vspace{-0.4em}\\

Education, $n$ (\%) & & \vspace{-0.7em}\\
\hspace{1em} College or above & 3431 (76.4) & 575 (64.0) \vspace{-0.8em}\\
\hspace{1em} High school or below & 1059 (23.5) & 175 (19.5) \vspace{-0.8em}\\
\hspace{1em} Missing & 0 (0.0) & 148 (16.5) \vspace{-0.4em}\\





Rectal Gonorrhea  (\%) & & \vspace{-0.7em}\\
\hspace{1em}Yes & 293 (6.5) & 33 (3.7) \vspace{-0.8em}\\
\hspace{1em} No & 4116 (91.7) & 654 (72.8) \vspace{-0.8em}\\
\hspace{1em} Missing   & 81 (1.8) & 211 (23.5) \vspace{-0.4em}\\
           
Rectal Chlamydia  (\%) & & \vspace{-0.7em}\\
\hspace{1em} Yes & 495 (11.0) & 65 (7.2)\vspace{-0.8em}\\
\hspace{1em} No & 3916 (87.2) & 623 (69.4) \vspace{-0.8em}\\
\hspace{1em} Missing   & 79 (1.8) & 210 (23.4) \vspace{-0.4em}\\

Rectal Gonorrhea or Chlamydia  (\%) & & \vspace{-0.7em}\\
\hspace{1em} Yes & 679 (15.1) & 82 (9.1)\vspace{-0.8em}\\
\hspace{1em} No & 3730 (83.1) & 603 (67.1)\vspace{-0.8em}\\
\hspace{1em} Missing   & 81 (1.8) & 213 (23.7) \\\bottomrule
\end{tabular}
}
\end{table}

\vspace{-0.5cm}\subsection{Gaps in existing statistical methods for external controls}\vspace{-0.3cm}

General criteria for evaluating the acceptability of historical or external controls were outlined by \cite{POCOCK1976175}. These criteria emphasize the importance of ensuring comparability in patient populations, study design, and trial conduct between the primary trial and the external data source. These foundational principles provide the basis for subsequent statistical approaches used to incorporate external controls.

A variety of methods have been developed to incorporate external control data under the assumption of no unmeasured confounding. This assumption requires that participants in the primary and external datasets be exchangeable at each level of the observed variables -- a strong assumption that cannot be empirically tested.
When this assumption holds, researchers can apply standard causal inference methods, including propensity score matching \citep{rubin1974estimating, dehejiaPSM, li2004PSM, liPSMreview, chenPSM, StuartAndRubinMatching}, inverse probability weighting \citep{Joffe2004,Chesnaye2021, syntheticreview}, G-computation \citep{snowdenGcomp, robins1986gform, hernan2020causal}, and double/debiased machine learning {and semi-parametric theory} \citep{ChernozhukovDML, belloni2017program, liu2025targeted}. Additionally, several methods have been developed specifically for external control settings \citep{liefficiencyexternalcontrol, valancius2024causal}. All of these approaches use observed covariates to enhance comparability between the primary and external populations and to adjust for potential confounding.

{Unmeasured confounders (e.g., local HIV prevalence, HIV viral load of sexual partners, local sexual network density, local behavior/cultural norms in HIV prevention) may influence both the risk of HIV infection and study participation, particularly since the two trials recruited participants from different geographic areas.}
In such cases, methods that rely on the assumption of no unmeasured confounding may produce biased estimates. When placebo data are available in both the primary and external datasets, researchers can formally assess comparability before combining data sources \citep{viele2014use, liu2022matching}. Sensitivity analysis frameworks, such as the one proposed by \cite{YiZhangDuYe+2023}, can help quantify the potential impact of unmeasured confounding.

When multiple historical or external trials are available, meta-analytic techniques can be used to model heterogeneity across studies and extrapolate findings to the target population \citep{networkmetaanalysis}. Bayesian borrowing methods offer a complementary strategy that allows the contribution of external data to be weighted according to its similarity to the primary trial population. Examples include power priors, elastic priors, and commensurate priors \citep{powerprior, elasticprior, commensurateprior, SCHMIDLImeta-prior}. However, Bayesian frameworks are often sensitive to model assumptions and may require substantial domain knowledge to guide prior specification.

Therefore, a gap remains in developing statistical methods that can address unmeasured confounding when leveraging external controls, particularly when the primary trial lacks a placebo arm. {Low HIV incidence rates may create additional statistical challenges, including reduced precision from the limited number of observed events and instability in estimation, where both the point estimates and their confidence intervals can fall outside the plausible range, e.g. $(0,100\%)$ for HIV cumulative incidence.} Finally, low HIV incidence rates may limit the power of the standard statistical analysis and create additional statistical challenges, which calls for most efficient methods

\vspace{-0.5cm}\subsection{Review of proximal inference methods}\vspace{-0.3cm}

Proximal causal inference methods were developed to address unmeasured confounding in observational studies by leveraging both negative control exposures (NCEs) and negative control outcomes (NCOs) \citep{MiaoProximal,shi2020multiply}. An NCO is a variable that is not causally affected by the treatment of interest, and an NCE is a variable that does not causally affect the outcome of interest. Importantly, both NCOs and NCEs must be associated with the unmeasured confounders that affect the primary exposure-outcome relationship.

Key foundational and review articles in this area include \cite{tchetgen2020introduction}, \cite{shi2020multiply}, and \cite{shi2020selective}. Building on this foundational work, \cite{cui2020semiparametric} proposed a semiparametric framework for identifying and estimating both the average treatment effect (ATE) and the average treatment effect on the treated (ATT). \cite{yingproximal} extended the estimation of ATE to time-to-event outcomes. More recently, researchers have developed regression-based approaches to proximal inference that employ generalized linear models for outcome modeling \citep{liu2024regression, li2024regression}.

Most relevant to our study, \cite{suproximalindirect} developed a proximal indirect comparison framework for settings where the treatment of interest is unavailable in the primary RCT but available in an external RCT. Their approach uses proxy variables and bridge functions to adjust for unmeasured, shifted effect modifiers between datasets, yielding a doubly robust and asymptotically normal estimator. However, their identification strategy requires that the two trials share a common third treatment arm, which is not the case in our study. 

\vspace{-0.5cm}
\subsection{Method development and HPTN 083 application}\vspace{-0.3cm}

The rich set of baseline covariates collected in HPTN 083 and the AMP study enables the identification of plausible negative control variables for proximal inference. {We define local HIV transmission environment risk as a collective measure of local HIV prevalence, viral load among partners, local sexual network density, and behavior/cultural norms in HIV prevention. It} is the primary unmeasured confounder of concern in comparing across trials, and our negative control variables serve as proxies for this factor. Baseline diagnoses of STIs, such as gonorrhea and chlamydia, serve as appropriate negative control outcomes because they share key behavioral and biological risk factors with HIV acquisition and are strongly correlated with {local HIV transmission environment risk}. {Moreover, STIs are not expected to directly affect which study an individual participates in, except through their association with shared confounders -- {STI eligibility criteria were similar in both trials, STIs did not affect trial catchment area selection}
and typically do not directly influence willingness to participate after adjusting for other covariates.} Geographic region is a reasonable negative control exposure because it correlates with the {local HIV transmission environment risk} but does not directly influence individual’s HIV risk. The availability of these strong proxy variables supports the application of proximal inference methods in this setting (see Section \ref{sec:notation} and Figure \ref{fig1}).

Our goal is to estimate the counterfactual cumulative HIV incidence under placebo in the HPTN 083 trial using external control data from the AMP study. To address unmeasured differences in HIV risk between the two trials, we extend the semiparametric proximal inference methods of \cite{cui2020semiparametric} and \cite{yingproximal} to handle external control integration with time-to-event outcomes. Additionally, to overcome the statistical challenges posed by low HIV incidence in both studies, we develop a novel two-stage regression-based proximal inference approach that is based on the Cox proportional hazards model.

The remainder of this paper is organized as follows. In Section \ref{sec:methods}, we introduce the notation and setup in Section \ref{sec:notation}, followed by a detailed description of the proposed proximal identification and inference procedures in Section \ref{sec:method}. In Section \ref{sec:realdata}, we present our analysis of the HPTN 083 dataset and discuss the results of the statistical inference. Finally, we conclude with a discussion in Section \ref{sec:discussion}.

\vspace{-0.5cm}\section{Methods}\label{sec:methods}
\vspace{-0.5cm}\subsection{Notation and setup}\label{sec:notation}\vspace{-0.3cm}
We denote the primary dataset with $J$ active treatment arms as $R = 0$ and the external control dataset as $R = 1$. Let $A$ represent the treatment assignment, where $A = 0$ indicates placebo and $A = 1, \dots, J$ represent active treatment arms. Unless otherwise specified, treatment assignment is partially determined by data source: specifically, $R = 1$ corresponds to the placebo arm ($A = 0$), and $R = 0$ corresponds to active treatments ($A = 1, \dots, J$). In our application, the primary dataset (HPTN 083) includes $J = 2$ active treatment arms: cabotegravir ($A = 1$) and TDF/FTC ($A = 2$).  Let $T = \min(T^*, C)$ denote the observed follow-up time, where $T^*$ is the time to HIV infection and $C$ is the censoring time. We define the event indicator as $\Delta = I(T^* \leq C)$, with $\Delta = 1$ indicating observed HIV infection and $\Delta = 0$ indicating right censoring. 
We use $\bm{X}$ to denote a vector of baseline covariates. {In addition, we observe a negative control outcome $W$ at baseline in both the primary and external datasets. Furthermore, a negative control exposure $Z$ is observed for the two-stage regression-based approach in both datasets but not required in the primary dataset for the semiparametric approaches. For each individual in the combined dataset of size $n$, we have $(R, T, \Delta, \bm{X}, Z, W)$.}

Let $T^*(a)$ denote the counterfactual time to HIV infection under treatment arm $a \in \{0, 1, \dots, J\}$.  We assume the Stable Unit Treatment Value Assumption (SUTVA) \citep{Rubin1980}, which implies the standard consistency assumption $T^*=T^*(A)$. Our primary target estimand is the counterfactual placebo risk in HPTN 083, defined as the probability of HIV infection by time $t$ under placebo: $P(T^*(0)\leq t\mid R=0)$. When $t$ corresponds to one year, this can also be interpreted as the one-year cumulative HIV incidence. We are also interested in the average treatment effect among the population in HPTN 083, defined by the contrast $P(T^*(a)\leq t\mid R=0) -P(T^*(0)\leq t \mid R=0) $ for $a=1,\cdots,J$. Under randomization within HPTN 083, the treated potential outcomes are identifiable, so $ P(T^*(a)\leq t\mid R=0) =P(T^*\leq t\mid R=0, A=a)$ for $a=1,\dots, J$. 

To address potential bias due to unmeasured confounding, we posit the existence of an unmeasured variable $U$ such that primary and external participants become comparable conditional on both $U$ and $\bmX$ 
\citep{rosenbaum1983sens,rosenbaum1983central, Imbens2003}. The presence of $U$ implies that comparability cannot be achieved by adjusting for $\bmX$ alone. For example, participants in HPTN 083 and the AMP study may differ in {local HIV transmission environment risk}, which affects HIV risk but is often not measured with sufficient spatial granularity or temporal alignment with the trial periods. As a result, the counterfactual outcome distributions under placebo may differ between datasets even after conditioning on $\bmX$: $ P(T^*(0)\leq t\mid \bmX, R=0) \neq P(T^*(0)\leq t\mid \bmX, R=1) $. This lack of exchangeability motivates the need for methods that can account for unmeasured confounding, such as proximal inference. 

Proximal inference leverages a negative control exposure $ Z $ and a negative control outcome $ W $ that satisfy the following key assumption \citep{tchetgen2020introduction}: 

\begin{assumption}[Negative controls] \label{assump: negative controls}
 	$ (Z, R)  \indep (T^*(0), W) \mid \bmX, U$. 
\end{assumption}

Assumption \ref{assump: negative controls} states that there exists an unmeasured confounder $U$ such that, conditional on both $U$ and the observed covariates $\bmX$, participants from the primary and external trials are comparable \citep{rosenbaum1983sens,rosenbaum1983central, Imbens2003}. Additionally, this assumption requires that the negative control exposure $Z$ does not have a direct effect on either the primary outcome $T^*(0)$ or the negative control outcome $W$, and that $W$ does not directly influence trial participation $R$. These conditions allow $Z$ and $W$ to serve as valid proxies for the unmeasured confounder $U$ in the proximal inference framework.

 \begin{figure}[!htbp]
 		\centering
 		\begin{tikzpicture}
 			\node[rectangle, draw, align=center] (r) at (0,0) {$R$\\ [-1em] {\footnotesize (study index)}};
 			\node[rectangle, draw, align=center] (z) [left =of r] {$Z$ \\ [-1em]{\footnotesize (NCE: coarse region)} };
 			\node[rectangle, draw, align=center] (y) [right =of r] {$T^*(0)$\\ [-1em] {\footnotesize (event time under placebo)}};
 			\node[rectangle, draw, align=center] (w) [right =of y] {$W$ \\[-1em] {(\footnotesize NCO: STI infection)}};
 			\node[rectangle, draw, align=center] (u) [above =of y,yshift=2cm] {$U$ ({\footnotesize unobserved: local HIV transmission environment risk}),\\[-0.5em] $\bmX$ {\footnotesize (observed: age, race, gender)}};
                \node[state, text width=2.5cm,align=center] (m2) [below =of u, xshift=-3.5cm]{local region};
 			\path (u) edge node[above]{} (y);
 			\path (u) edge node[above]{} (w);
 			\path (m2) edge node[above]{} (r);
 			\path (m2) edge node[above]{} (u);
                \path (m2) edge node[above]{} (z);
 		\end{tikzpicture} 
 	\caption{Causal DAG depicting the assumed causal relationships underpinning Assumption \ref{assump: negative controls}. The unmeasured confounder $U$, such as {local HIV transmission environment risk}, and observed covariates $\bmX$ jointly affect the potential HIV outcome under placebo $T^*(0)$, the negative control outcome $W$ (e.g., STI status), and the study indicator $R$. The coarse regional variable $Z$ (e.g., Latin America vs. non-Latin America) serves as a proxy for $U$ through its dependence on local region. Under this structure, $Z$ and $W$ satisfy the conditions for valid negative control exposure and outcome, respectively.  \label{fig1}}
 \end{figure}
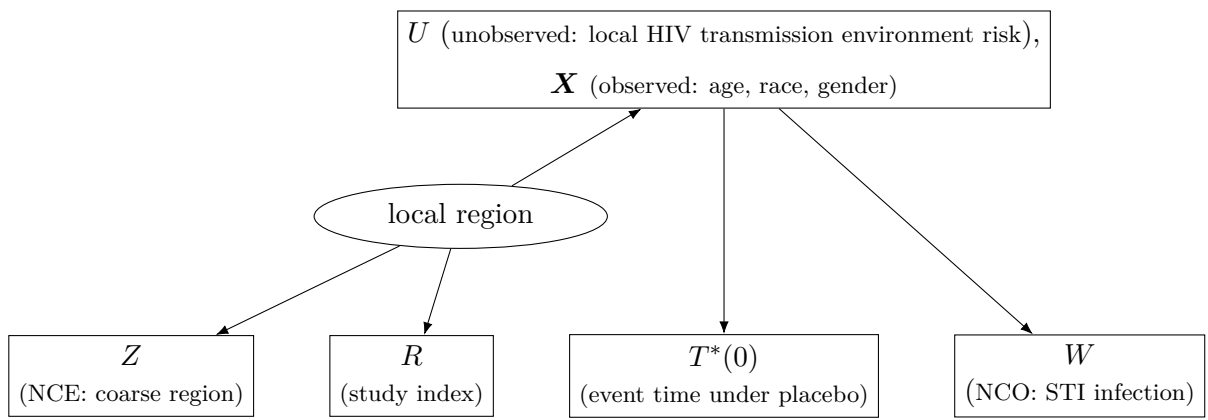

In our application, the assumed causal relationships among the variables are depicted in the directed acyclic graph (DAG) in Figure \ref{fig1}. A key unmeasured confounder $U$ is the {local HIV transmission environment risk} at a granular geographic level. This granular location variable (labeled “local region” in the DAG) influences both the coarse regional variable $Z$ (e.g., Latin America vs. non-Latin America) and the unmeasured HIV prevalence $U$, making $Z$ a proxy for $U$. Geographic location also determines study accessibility, as enrollment is often site-specific.  Thus, the unmeasured HIV prevalence $U$ is associated with study assignment $R$ through geographic region and also directly affects an individual’s risk of HIV infection, making it a source of unmeasured confounding. Importantly, the coarse regional variable $Z$ does not directly influence HIV risk beyond its association through $U$ and covariates $\bmX$, which include individual-level HIV risk factors. Furthermore, the negative control outcome $W$ (e.g. baseline STI infections) can share similar risk factors ($U,\bmX$) as HIV infections, but it does not directly influence the study index $R$ because it does not affect trial catchment area selection, is not part of either trial's inclusion-exclusion criteria, and typically does not directly influence willingness to participate after adjusting for other covariates.

\vspace{-0.5cm}\subsection{A proximal causal inference approach}\label{sec:method}\vspace{-0.3cm}
\subsubsection{Method 1: Semiparametric IPCW estimator}\label{sec:semiparametric}\vspace{-0.3cm}

In this section, we develop a semiparametric identification and estimation strategy based on inverse probability of censoring weighting (IPCW) within the proximal causal inference framework. We begin by introducing two key assumptions -- positivity and completeness -- that are commonly used in the proximal inference literature \citep{MiaoProximal,  cui2020semiparametric} and can be readily adapted to our setting.

\begin{assumption}[Positivity] \label{assump: positivity}
	For some constant $c>0$, $c< P(R=1\mid \bmX, U)<1-c$ {holds $P_{U,\bm X}(\cdot|R=0)$-almost surely.} 
\end{assumption}


This condition ensures that, for every $(\bmX,U)$ combination with positive probability in the primary dataset $(R=0)$, there is a positive probability of observing that $(\bmX,U)$ combination in the external dataset $(R=1)$. {Note that the covariates $\bmX$ should be harmonized as needed (e.g., via re-categorization) to ensure that all levels of $\bmX$ in the primary dataset should also exist in the external dataset.} 


\begin{assumption}[Completeness] \label{assump: completeness} {For any square integrable function $g$ and for all $\bmx$:\\
(a) $g(u)=0 $ for all $u$ whenever $E[g(U)\mid Z=z, R=1, \bmX = \bmx] =0$ for all $z$. \\
(b) $g(u)=0 $ for all $u$ whenever $E[g(U)\mid W=w, R=1, \bmX = \bmx] =0$ for all $w$.}

\end{assumption}

{The validity of the proximal inference approach crucially relies on whether the identified negative control variables are sufficiently informative about the unmeasured confounder. Assumption \ref{assump: completeness} provides a sufficient condition to ensure identification in Theorem \ref{theo: censoring}. This assumption can be easier to interpret when $ U, Z$, and $W $ are categorical. In this case, Assumption \ref{assump: completeness} requires that both $ Z $ and $ W $ have at least as many categories as $ U $, and that variation in $U$ can be
recovered from variation in $Z$ and $W$. More specifically, within $R=1$ and for every possible value of $\bmx$, different values of $U$ must correspond to distinct distributions of both $W$ and $Z$
\citep{MiaoProximal,shi2020selective}. When $U$ is continuous, Assumption \ref{assump: completeness} requires that within the external dataset and for any level of $\bmX$, no function of $U$ can be orthogonal to the information provided by $W$ and $Z$ \cite{olivas2025proximal}. In other words, $W$ and $Z$ need to be sufficiently predictive of $U$, $U^2$, $U^3$, and more generally, all possible variation in $U$. See \cite{MiaoProximal} for discussions of various parametric, semiparametric and non-parametric models that satisfy this assumption. For strategies to relax this assumption, we refer readers to \citet{bennett2022inference}.} 

To fully account for unmeasured confounding due to a continuous $U$ (e.g., {local HIV transmission environment risk}), the semi-parametric approach requires continuous negative control variables that are strongly associated with $U$. Such variables are currently lacking in our data, as both $Z$ (region) and $W$ (baseline STI infections) are binary. Therefore, we modify Assumption \ref{assump: negative controls} by assuming that controlling for a dichotomized version of {local HIV transmission environment risk} (high vs low) suffices to address unmeasured confounding.


\begin{assumptionone*}[Negative controls and binary unmeasured confounders]\label{assump:negative controls - semi}
$(Z, R)\indep (T^*(0), W)\mid \bmX, U_b$, where $U_b$ is a dichotomized version of $U$.
\end{assumptionone*} 

In Supplement \ref{supp:sensi_binU}, we analytically assess the magnitude of bias and present simulation results for scenarios where a binary $U_b$ and observed covariates $\bmX$ fail to fully adjust for confounding between $(Z, R)$ and $(T^*(0), W)$, thus violating Assumption \ref{assump: negative controls}$^*$. In this case, we find that stronger proxy variables $Z$ for the unmeasured confounder $U$ yield reduced bias and smaller standard errors.

{In addition, we account for right censoring since the primary endpoint is time to HIV infection. Censoring arises from two sources: (1) most cases are administratively censored at the study end where participants have not experienced the event by the end of the study period, and (2) the remainder reflect loss to follow-up (e.g., dropout). We assume censoring at random in the external study, allowing the censoring mechanism to depend on observed covariates $\bm X$ and the negative control exposure $Z$, but not on the unmeasured confounder $U$ or the negative control outcome $W$. This assumption is standard and is expected to hold approximately when $\bm X$ includes a sufficiently rich set of participant characteristics.

} 

\begin{assumption}[Censoring] \label{assump: censoring}
Let \( \lambda_C(t \mid \cdot) \) denote the conditional hazard function for censoring at time \( t \). We assume
	$\lambda_C(t\mid T>t, Z, \bmX, U, W, T^*, R=1)=\lambda_C(t\mid T>t, Z, \bmX, R=1).$
\end{assumption}

Theorem \ref{theo: censoring} presents the identification result for the proposed semiparametric IPCW method, extending the work of \cite{cui2020semiparametric} to account for right censoring.

\begin{theorem}[Semiparametric IPCW identification] \label{theo: censoring}
{Suppose Assumptions 1$^*$, \ref{assump: positivity}-\ref{assump: censoring}, and that there exist square-integrable functions $ h $ and $ q $ that satisfy the following integral equations almost surely:
 	\begin{align}
 		&E\left[\frac{\Delta I(T^*\leq t)}{P(C>T^*\mid Z, \bmX, R=1)}-h(W,\bmX)\mid Z,\bmX, R=1\right]=0 , \label{eq: h.c}\\
 		&\frac{P(R=0\mid W,\bmX)}{P(R=1\mid W,\bmX)}= E[ q(Z,\bmX) \mid W,\bmX, R=1]. \label{eq: q}
 	\end{align}
 Then, $P(T^*(0)\leq t\mid R=0)$ is nonparametrically identified and can be expressed in either of the following two forms: 
 	\begin{align}
 	 &P(T^*(0)\leq t\mid R=0)= E[ h(W,\bmX) \mid R=0], \label{eq:outcomebridge}\\
 	 &P(T^*(0)\leq t\mid R=0)= \frac{1}{P(R=0)}E \left[\frac{R {q}(Z,\bmX) \Delta I(T^*\leq t)}{P(C>T^*\mid Z, \bmX, R=1)}\right].\label{eq:propensityest}
   \end{align}
Furthermore, $P(T^*(0)\leq t\mid R=0)$ can be identified through the following augmented form:
        \begin{align}\label{eq:doublerobust}
 	 P(T^*(0)\leq t\mid R=0) & = \frac{1}{P(R=0)}  E \bigg[R q(Z,\bmX) \bigg
     \{\frac{\Delta I(T^*\leq t)}{P(C>T^*\mid Z, \bmX, R=1)} - h(W,\bmX)  \bigg\} \nonumber\\
     &\ \ \ +(1-R) h(W,X) \bigg],
   \end{align}
which remains valid as long as either $h$ satisfies \eqref{eq: h.c} or $q$ satisfies \eqref{eq: q}. 
}
 \end{theorem}

In the proximal inference literature, the functions $h$ and $q$ are referred to as the outcome bridge function and treatment bridge function, respectively. In Theorem \ref{theo: censoring}, we do not require $h$ and $q$ to be unique, which relaxes the assumptions in \cite{cui2020semiparametric}. Any pair of functions that satisfy \eqref{eq: h.c}--\eqref{eq: q} identify the same target parameter.

Directly solving equations \eqref{eq: h.c} and \eqref{eq: q} is challenging, as they are integral equations whose solutions may be non-unique and numerically unstable. How to properly address these challenges is still an active research area and we refer interested readers to \cite{kress1989linear,dikkala2020minimax,ghassami2022minimax}, and \cite{bennett2023source}. In the remainder of this section, we briefly describe our estimation procedures for $h$ and $q$ as implied by Theorem \ref{theo: censoring}, {assuming each takes a parametric form}. 

To estimate $h$, we first need to estimate the conditional survival probability $P(C>T^*\mid Z, \bmX, R=1)$. This can be done using the external dataset by fitting a Cox proportional hazards model for the censoring time, treating censoring as the event and including $Z,\bmX$ as covariates. {The estimated survival probability is: $
    \hat P(C>T^*|Z,\bmX, R=1) = \exp\left\{ -\hat H(T^*; Z,\bmX, R = 1)\right\}$, 
where $\hat H (T^*;\cdot)$ is the estimated cumulative hazard function up to time $T^*$. Note that when $\Delta=1$, we have the observed time $T=T^*$.} Next, we reformulate Equation \eqref{eq: h.c} as: 
\begin{align}
    E\left[g_h(Z,\bmX)\left[\frac{\Delta I(T^*\leq t)}{{P}(C>T^* \mid Z,\bmX, R=1)}-h(W,\bmX)\right]\mid R=1\right]=0 
    \label{eq: h.c marginal}
\end{align}
for any arbitrary function $g_h(Z,\bmX)$. We parameterize $h(W, \bmX)$ and select $g_h(Z, \bmX)$ as a set of transformations of $Z$ and $\bmX$ so that the number of moment equations in \eqref{eq: h.c marginal} matches the number of parameters in $h(W, \bmX)$. This ensures that the system of equations is just-identified and can be solved using the external dataset. 

A similar strategy is used to estimate $q$. By multiplying \eqref{eq: q} by an arbitrary function $g_q(W, \bmX)$, we obtain $$E\left[ g_q(W,\bmX) \left[q(Z,\bmX) - \frac{P(R=0\mid W,\bmX)}{P(R=1\mid W,\bmX)}\right]\mid R=1\right]=0,$$
where $P(R=1\mid W, X)$ can be estimated from a logistic regression. We parameterize $q(Z, \bmX)$ and choose $g_q(W, \bmX)$ as a set of transformations of $W$ and $\bmX$ to solve for the parameters in $q(Z, \bmX)$.

Once either $h$ or $q$ has been estimated, we apply the identification results from Theorem \ref{theo: censoring} to estimate $P(T^*(0)\leq t|R=0)$ via the outcome bridge form \eqref{eq:outcomebridge}, the treatment bridge form \eqref{eq:propensityest}, or the augmented (doubly robust) form \eqref{eq:doublerobust} that incorporates both bridge functions. { 
One caveat is that IPCW estimates may fall outside the (0,1) range when event rates are low or proxies are weak. Standard errors and confidence intervals can be obtained via nonparametric bootstrap or a sandwich variance estimator if a parametric censoring model (e.g., exponential distribution) is appropriate. For better finite-sample performance, we apply a $\log (-\log(\cdot))$ transformation, construct confidence intervals on the transformed scale, and back-transform to the original scale.}  

Supplement \ref{supp:equiv_hq} provides additional results on the semiparametric approach. We show that if both the outcome and treatment bridge functions are estimated nonparametrically -- and the treatment bridge is estimated by solving $E\left[ g_q(W,\bmX) \left[I(R=1)q(Z,\bmX) - I(R=0)\right]\right]=0$
instead of using Equation \eqref{eq: q} -- the resulting estimators from both bridge functions are identical.
However, this equivalence generally does not hold when parametric models are used. We also discuss the implications when either $Z$ or $W$ is unrelated to $U$, as indicated by $Z\indep  W\mid \bmX$. To guard against this scenario, we recommend 
assessing the conditional association between $Z$ and $W$ given $\bmX$ before conducting the main analysis.

\vspace{-0.5cm}\subsubsection{Method 2: Regression-based two-stage estimator} \label{sec:twostage}\vspace{-0.3cm}

The semiparametric estimator described above relies on minimal assumptions about the underlying data-generating process, but in HIV prevention trials where event rates are typically low, regression-based methods that offer greater efficiency may be more practical. Previous regression-based proximal causal inference methods have been proposed by \cite{liu2024regression} for generalized linear models and by \cite{li2024regression} for right-censored time-to-event outcomes using an additive hazards framework. Since Cox proportional hazards models \citep{cox1972regression} are more commonly used in HIV prevention trials, we develop a new regression-based proximal inference approach based on the Cox model, specifically tailored to settings with low-incidence outcomes.

For the regression-based approach, we introduce a new assumption that reflects the low incidence of HIV typically observed in prevention trials. This rare event assumption helps address the non-collapsibility of the Cox model. 

\begin{assumption}[Rare Event]\label{assump:rare event}
    For some pre-specified time point $t_0$, the event-free probability under placebo satisfies $P(T^*(0)>t_0|U,\bmX) \approx 1$, reflecting the low incidence of the event.
\end{assumption}

Under Assumption \ref{assump:rare event}, the hazard function under placebo can be approximated by the density function, which is collapsible \citep{tchetgen2015ivsurvival}. Specifically, for $t\leq t_0$, $ \lambda^{(0)}(t\mid U,\bmX) \approx f^{(0)}(t\mid U,\bmX)$, where $\lambda^{(0)}(t\mid \cdot)$ and $f^{(0)}(t\mid \cdot)$ denote the conditional hazard and density functions under placebo. In our application, we set $t_0=1$ year, a natural time unit for HIV prevention trials, justified by typically low 1-year cumulative incidence under placebo. For example, the AMP study reported a one-year cumulative incidence of 2.98\% under placebo, corresponding to {a one-year probability of no HIV infection of 0.97.}

Next, we specify a distributional assumption for the unmeasured confounder $U$. 

\begin{assumption}[Unmeasured confounder]\label{assump:latent variable}
    We assume that $U = E(U|R,Z,\bmX) + \epsilon$, where $\epsilon \indep (R,Z,\bmX)$, $E(\epsilon) = 0$, {and $E(U\mid R=1, Z, \bmX)\neq E(U\mid R=1
    , \bmX)$ almost surely.}
\end{assumption}

Assumption \ref{assump:latent variable} is appropriate when $U$ is continuous, as it implies a mean-independent error structure. It does not hold when $U$ is binary. We also explored a modified version of the assumption for binary $U$ based on an alternative model, but found that estimation becomes unstable due to the large number of parameters involved. 
Therefore, when applying the two-stage approach, it should be understood that $U$ is assumed to be continuous. Note that underlying data-generating processes can satisfy the identification assumptions of both the IPCW approach (with binary
$U$) and the two-stage approach (with continuous $U$); we provide an example in Supplement \ref{supp:sim}. In addition, we assume that the negative control exposure $Z$ is related to $U$ conditional on $\bmX$ in the external study. Notably, the completeness assumption (Assumption \ref{assump: completeness}) is not required in the two-stage approach and is implicitly replaced by the standard full-rank conditions in generalized linear model, which is why we can use binary $Z$ and $W$ to address unmeasured confounding even if $U$ is continuous.

Next, we specify a log-linear model for the binary negative control outcome $W$ and a Cox model for the counterfactual event time $T^*(0)$. 

\begin{assumption}[Models for $W$ and $T^*(0)$]\label{assump:NCO structure}
(a) We assume that the negative control outcome \( W \) follows a log-linear model: $E(W|U,\bmX) =  \exp(\beta_{W0} + \beta_{Wx}\bmX + \beta_{Wu}U)$, where \( \beta_{W0} \), \( \beta_{Wx} \), and \( \beta_{Wu} \) are model parameters, {and $\beta_{Wu}\neq 0$.} 
(b) We assume that the counterfactual event time $T^*(0)$ follows a Cox model: $\lambda^{(0)}(t \mid U,\bmX) = \lambda_0(t)\exp(\beta_{Tu} U + \beta_{Tx}\bmX) $, where \( \lambda_0(t) \) is the baseline hazard function and \( \beta_{Tu} \), \( \beta_{Tx} \) are model parameters.
\end{assumption}

Lastly, we also impose Assumption \ref{assump: negative controls} for the negative controls and Assumption \ref{assump: censoring} for censoring at random in the external data. Theorem \ref{theo:twostage} then establishes the main identification results for the regression-based approach.

\begin{theorem}\label{theo:twostage}
Under Assumptions \ref{assump: negative controls} and \ref{assump: censoring}-\ref{assump:NCO structure}, for any  $t\leq t_0$, we have 
\begin{align}
    \lambda^{(0)}(t \mid R, Z, \bmX) \approx \tilde{\lambda}_0(t) \exp\Big[ \big\{ \tilde\beta_{Tx} \bmX + \tilde\beta_{Tw} \log\big( E(W \mid R, Z, \bmX) \big) \big\} \Big], \label{eq: two-stage main identification}
\end{align}
where $\tilde \lambda_0(t)$, $ \tilde\beta_{Tx}$ and $ \tilde \beta_{Tw}$ are parameters that are identifiable using data from the external study ($R=1$). 
As a result, the counterfactual hazard function for the primary study can be identified as  $ \lambda^{(0)}(t|R=0,Z,\bmX) = \tilde{\lambda}_0(t) \exp\Big[ \big\{ \tilde\beta_{Tx} \bmX + \tilde\beta_{Tw} \log\big( E(W \mid R=0, Z, \bmX) \big) \big\} \Big]$.  
\end{theorem}

The intuition behind our approach is illustrated in the following derivations. Let $t\leq t_0$. Under Assumptions \ref{assump: negative controls}, \ref{assump:rare event}, and \ref{assump:NCO structure}, we have  
\begin{align*}
	f^{(0)}(t\mid R,Z, U, \bmX)&\approx\lambda^{(0)}(t \mid R, Z, U,\bmX) = \lambda_0(t)\exp(\beta_{Tu} U + \beta_{Tx}\bmX), \\
	E(W \mid R, Z,  U,\bmX) &=\exp(\beta_{W0}   + \beta_{Wx}\bmX+ \beta_{Wu} U). 
\end{align*} 
Next, marginalizing over the distribution of $U$ conditional on $(R, Z, \bmX)$, and applying Assumption \ref{assump:latent variable}, we obtain:
\begin{align*}
	f^{(0)}(t \mid R, Z,\bmX) &= E\{ \lambda_0(t)  \exp(\beta_{Tu}\epsilon) \}\exp(\beta_{Tx}\bmX+\beta_{Tu}E(U\mid R, Z,\bmX))\\
	E(W \mid R, Z,\bmX) &= E\{ \exp(\beta_{W0}  + \beta_{Wu}\epsilon) \}\exp( \beta_{Wx}\bmX+\beta_{Wu}E(U\mid R, Z,\bmX) )
\end{align*} 
Hence, there exist $ \tilde \beta_{T0}, \tilde \beta_{Tu},  \tilde \beta_{Tx} $, which are shared across primary and external datasets,  such that $		\lambda^{(0)}(t\mid R, Z, \bmX)\approx 
        f^{(0)}(t\mid R,Z,\bmX) = \tilde \lambda_0(t) \exp \big\{\tilde\beta_{Tx} \bmX + \tilde\beta_{Tw} \log(E(W|R,Z,\bmX))\big\},$ where the approximation is by $P(T^*(0)>t_0|R, Z, \bmX)=E\{ P(T^*(0)>t_0|U,\bmX) \mid R, Z, \bmX \}  \approx 1$. This completes the argument underlying Theorem \ref{theo:twostage}.

For implementation, we first estimate $E(W\mid R, Z, \bmX)$ using data from both the primary and external studies. Since only the external dataset contains observations under placebo, and the model coefficients are shared across datasets, we use the external data to fit the Cox model in equation \eqref{eq: two-stage main identification}. Specifically, we regress the observed event times on $\bmX$ and $\log E(W\mid R=1, Z, \bmX)$ as regressors in the Cox regression. This yields estimates of the parameters $\tilde\lambda_0(t), \tilde\beta_{Tx}, \tilde\beta_{Tw}$. These parameters are identifiable because, under the assumed conditions, $E(W\mid R=1, Z, \bmX)$ varies with $Z$, providing sufficient variation for estimation. Once the parameters are estimated, we can compute $\lambda^{(0)}(t\mid R=0,Z,\bmX) $ according to Theorem \ref{theo:twostage}. {This in turn yields $P(T^*(0)\leq t|R=0, Z, \bmX)$ and thus $P(T^*(0)\leq t|R=0)= E\{P(T^*(0)\leq t|R=0, Z, \bmX)\mid R=0 \}$.} When $\lambda_0(t)$ is assumed to be time-invariant, all model parameters -- including the parameter of interest and nuisance parameters -- can be stacked, and their estimation can be formulated through estimating equations.

{For statistical inference, one may use the nonparametric bootstrap. Alternatively, when $\lambda_0(t)$ is assumed time-invariant, a sandwich variance estimator can be applied. {Similarly, the parameter of interest is {transformed using} $\log (-\log(\cdot))$ to improve asymptotic approximation.} In our empirical results, the sandwich approach yields more precise variance estimates than the bootstrap when events are rare.}

We evaluate the finite sample performance of all proposed proximal inference estimators through comprehensive simulation studies presented in Supplement \ref{supp:sim}.


\vspace{-0.7cm}\section{Case Study: Application to HIV Prevention Trials}\label{sec:realdata}\vspace{-0.4cm}

\subsection{Estimating counterfactual HIV incidence under placebo}\vspace{-0.3cm}
We applied our method to HPTN 083, using the placebo arm of the AMP study as external control data to estimate the counterfactual one-year cumulative HIV incidence under placebo in HPTN 083. Given the similar HIV incidence rates across the three AMP study arms, we also present results using all three arms combined as a single external control dataset in the Supplement Section \ref{supp:allarms}.

Covariates $\bmX$ in the adjustment set included age (18-20, 21-30, $\geq$30 years), gender (male, other), and race (Black, White, other). {Note that the Cox model only showed significant associations between age and race and the outcome, but we additionally adjusted for gender in this primary analysis because of its relevance to HIV infection in the literature.} Geographic region was used as the NCE, defined as Latin America ($Z = 0$) vs. non-Latin America ($Z = 1$). {Baseline sexually transmitted infection (STI) status based on laboratory testing served as the NCO, with $W = 1$ indicating a positive STI result Table \ref{table:table1}. We evaluated three NCO options: (1) baseline rectal gonorrhea infection, (2) baseline rectal chlamydia infection, and (3) either baseline rectal gonorrhea or chlamydia infection. To evaluate the validity of the NCO, we fitted logistic regressions of $W$ on $Z$ and $\bmX$. The association between $W$ and $Z$ {given $\bmX$} when using baseline rectal chlamydia as the NCO was statistically insignificant (OR $=0.846, p=0.068$), indicating it is not a suitable proxy for the unmeasured confounder. Both baseline rectal gonorrhea and the combined STI showed statistically significant associations with geographic region {conditional on $\bmX$} (OR $= 0.605, p<0.001$ and OR $=0.740, p<0.001$, respectively). 
However, given the low prevalence of baseline gonorrhea infection in both datasets (6.5\% in HPTN 083 and 3.7\% in AMP), we present the primary analysis with the combined STI indicator as the NCO.  Results for individual STI components are reported in Supplement \ref{supp:chlamNCO}.}

We applied both proposed approaches to estimate the counterfactual HIV incidence under placebo. The first was a semiparametric IPCW method (Section \ref{sec:semiparametric}) with a linear outcome bridge function, $h(W, \bmX) = \beta_{h0} + \beta_{hW} W + \beta_{hX} \bmX$, and a linear {treatment bridge} function, $q(Z, \bmX) = \beta_{q0} + \beta_{qZ} Z + \beta_{qX} \bmX$. A doubly robust estimator was then constructed by combining these bridge functions. 
The second was a regression-based two-stage method (Section \ref{sec:twostage}) that assumed a time-invariant baseline hazard and was applied to both the full follow-up data and a truncated dataset where all participants were administratively censored at one year, corresponding to $t_0 = 1$ year in Theorem \ref{theo:twostage} to satisfy Assumption \ref{assump:rare event}. Specifically, we estimated the survival function on the external data ($R=1$) assuming an exponential distribution, then used the fitted function to predict survival probabilities in the primary HPTN 083 dataset.   
For comparison, we report results from naïve regression models that adjust for $\bmX$ alone or for $\bmX$, $Z$, and $W$ jointly. Estimates and confidence intervals for all methods were obtained by the \texttt{gmm} package in \textsf{R}. To improve asymptotic approximation, we used the transformed the $\log(-\log(\cdot))$ of the target parameter.

{In our primary analysis,  participants with missing NCO values (fewer than 10\% missing for each NCO) were excluded as shown in Figure \ref{fig:flow}, with results in Figure \ref{fig:forest_est}. We also explored alternative covariate adjustment sets in Supplement \ref{supp:covset}: (1) age and race;  (2) age, race, gender, ethnicity, and education level (college or above vs. high school or below). Adjusting for fewer covariates yields similar estimates, while adjusting more covariates produces larger estimates and wider confidence intervals because of the low event rates.}


\begin{figure}[!htbp]
    \centering
    \includegraphics[width=\linewidth]{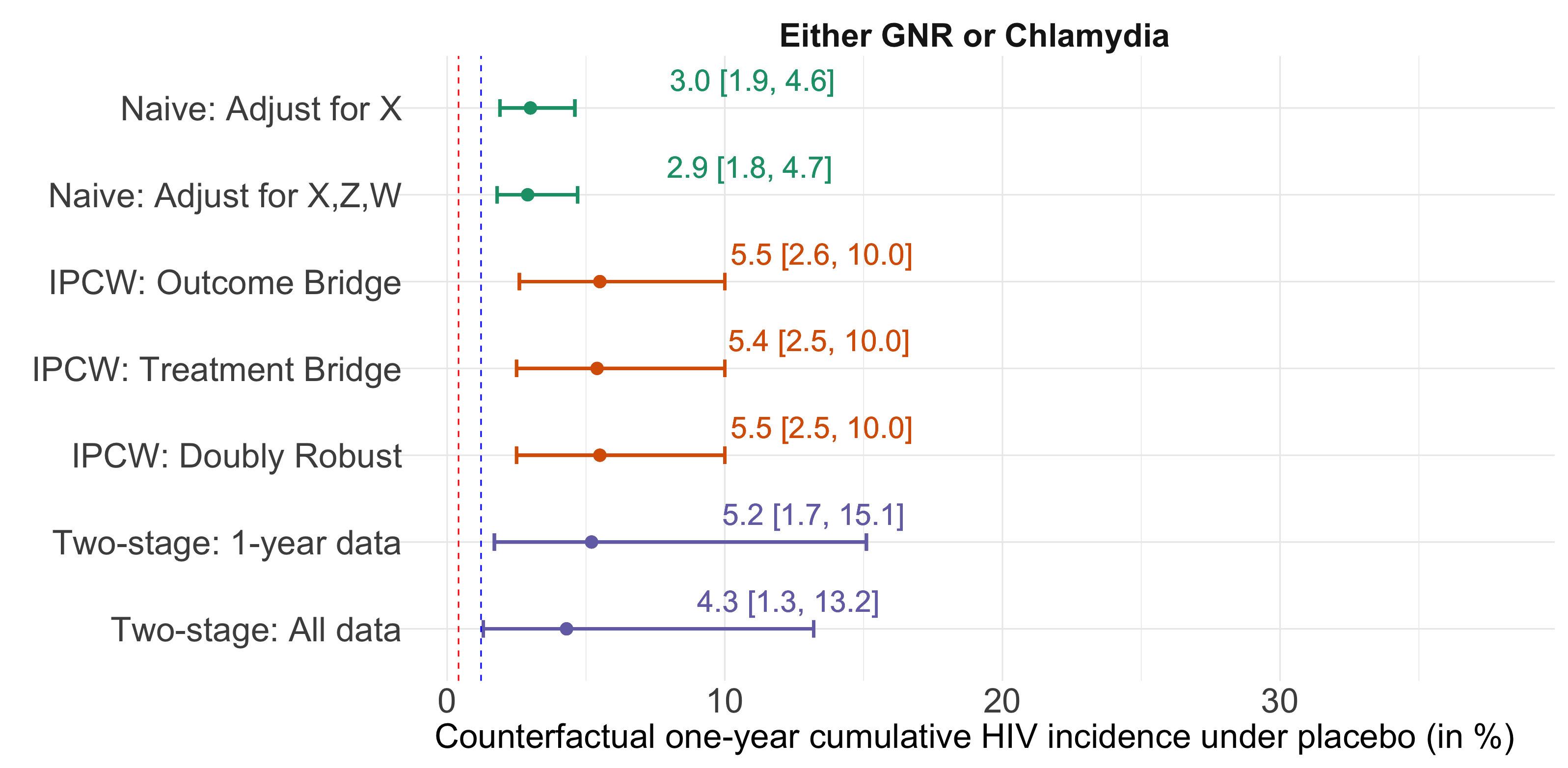}
    \caption{Estimated counterfactual one-year HIV cumulative incidence under placebo (in \%) for HPTN 083 using different methods and either baseline gonorrhea or chlamydia infection as the NCO. The red dashed line indicates the observed one-year cumulative incidence in the cabotegravir arm (0.41\%), and the blue dashed line indicates the observed one-year cumulative incidence in the TDF/FTC arm (1.22\%). 
    }
    \label{fig:forest_est}
\end{figure}

Across all proximal causal inference approaches, counterfactual one-year cumulative HIV incidence under placebo was consistent. The three semiparametric IPCW estimators (outcome bridge, treatment bridge, and doubly robust) and the two two-stage regression-based methods yielded point estimates, ranging from 4.3 to 5.5 per 100 person-years. As expected, the two-stage regression-based approach applied to full follow-up data produced narrower confidence intervals than the one-year analysis, due to more observed events. 

{In contrast, the two naïve regression models produced lower estimates of counterfactual HIV cumulative incidence under placebo, likely reflecting residual bias from unmeasured confounding. The naïve model yielded point estimates of 3.0 and 2.9 per 100 person-years adjusting only for $\bmX$ or for $\bmX$ as well as the negative control variables $Z$ and $W$. 
} 
Furthermore, the naïve regression models produce narrower confidence intervals because they assume no unmeasured confounding exists and that outcome models based on observed covariates from external control data can be directly applied to the primary dataset. In contrast, proximal learning recognizes that these outcome models may be incompatible across datasets due to unmeasured confounding and use negative controls to address them, therefore introducing additional uncertainty that results in wider confidence intervals.

\vspace{-0.5cm}
\subsection{Statistical inference on absolute efficacy} \label{sec:inference}
\vspace{-0.3cm}

Next, we used the counterfactual one-year HIV cumulative incidence under placebo, estimated from the proximal causal inference approaches, to evaluate the absolute efficacy of both treatments in HPTN 083. For comparison, the observed one-year HIV cumulative incidence in HPTN 083 was 0.41\% (95\% CI: [0.2, 0.7]) in the cabotegravir arm and 1.22\% (95\% CI: [0.9, 1.7]) in the TDF/FTC arm. 

To assess statistical significance, we computed Wald test statistics to determine whether significant differences existed in one-year HIV cumulative incidence between each active treatment arm and the estimated placebo arm. The Wald statistics for all approaches were calculated as follows:
$$
T_{\rm CAB}=\frac{\hat\beta_{\rm CAB.log}-\hat\beta_{\rm placebo.log}}{SE(\hat{\beta}_{\rm CAB.log}-\hat{\beta}_{\rm placebo.log})}, \hspace{2cm} T_{\rm TDF}=\frac{\hat\beta_{\rm TDF.log}-\hat\beta_{\rm placebo.log}}{SE(\hat{\beta}_{\rm TDF.log}-\hat{\beta}_{\rm placebo.log})}.
$$
{In the formula above, $\hat{\beta}_{*.\rm log}=\log(-\log(\hat{\beta}_*))$ denotes the complementary log-log transformed estimates of one-year cumulative incidence in each treatment arm ($*= \rm{TDF}$ or $\rm CAB$) in HPTN 083, $\hat{\beta}_{\rm placebo.log}=\log(-\log(\hat{\beta}_{\rm placebo}))$ denotes the transformed counterfactual placebo estimates from either IPCW or two-stage regression-based methods. Note that $\hat{\beta}_{\rm *.log}$ and $\hat{\beta}_{\rm placebo.log}$ are correlated because $\hat{\beta}_{\rm placebo.log}$ is estimated using participants from both treatment arms.} To account for this correlation, we directly estimate
$SE(\hat{\beta}_{\rm *.log}-\hat{\beta}_{\rm placebo.log})$ using sandwich variance estimators for IPCW and two-stage regression-based approaches.

The results are presented in Table \ref{tab:treatment_efficacy_analysis_placebo}. 
{Our results demonstrated that cabotegravir has} statistically significant superiority over the counterfactual placebo arm across all proximal inference methods (all p-values $\leq$ 0.001). 
All approaches also demonstrated statistically significant superiority of {TDF/FTC over placebo with estimated relative efficacy being 71.6\%-77.8\% although with larger p-values.
This pattern aligns with expectations, given that the AMP study placebo arm maintained non-negligible PrEP use (detectable concentrations of TDF/FTC in 39.0\% of the samples in AMP placebo arm vs. 72.3\% in HPTN 083 TDF/FTC arm, as discussed in Section \ref{amp_external_intro}. These findings provide robust evidence for cabotegravir efficacy and moderate evidence for TDF/FTC efficacy when compared to counterfactual placebo conditions. } Despite limitations, the results demonstrate that the proposed methodology can serve as a viable analytical framework for evaluating placebo-controlled efficacy in future trials without concurrent placebo control arms, particularly for investigational agents with comparable efficacy to cabotegravir.

\begin{table}[!htbp]
\centering
\caption{Estimated relative effectiveness (one minus ratio of the one-year cumulative HIV incidences of active treatment arms over placebo), estimated absolute efficacy (difference in one-year cumulative HIV incidence), Wald test statistics, and p-values comparing active treatment arms to the estimated placebo arm. Statistical comparisons are presented for cabotegravir vs. placebo and TDF/FTC vs. placebo across different analytical approaches.}
\label{tab:treatment_efficacy_analysis_placebo}
\resizebox{\textwidth}{!}{
\begin{tabular}{llccccccccc}
\toprule
{NCO} & {Method} & {Specification} & \multicolumn{4}{c}{{Cabotegravir}} & \multicolumn{4}{c}{{TDF/FTC}} \\ & & & {Rel. Efficacy } & {Abs. Efficacy} & {Test Stat.} & {p-value} & {Rel. Efficacy} & {Abs. Efficacy } & {Test Stat.} & {p-value}\\ 
\midrule 

& Naïve & Adjusted for $X$     
 & 0.863 & 0.026 & -5.593 & $<0.001$  
 & 0.593 & 0.017 & -3.111 & 0.002 \\\midrule
 
 \multirow{7}{*}{\centering {Either STI}}
 & Naïve & Adjusted for $X,Z,W$     
 & 0.859 & 0.025 & -5.205 & $<0.001$  
 & 0.579 & 0.017 & -2.933 & 0.003 \\

 & IPCW & Outcome bridge     
 & 0.925 & 0.051 & -4.960 & $<0.001$  
 & 0.778 & 0.043 & -3.360 & 0.001 \\

 & IPCW & Treatment bridge     
 & 0.924 & 0.050 & -4.870 & $<0.001$  
 & 0.774 & 0.042 & -3.284 & 0.001 \\

 & IPCW & Doubly robust     
 & 0.925 & 0.051 & -4.945 & $<0.001$  
 & 0.778 & 0.043 & -3.346 & 0.001 \\

 & Two-stage & 1-year data     
 & 0.921 & 0.048 & -3.692 & $<0.001$  
 & 0.765 & 0.040 & -2.252 & 0.023 \\

 & Two-stage & All data     
 & 0.905 & 0.039 & -3.731 & $<0.001$  
 & 0.716 & 0.031 & -2.202 & 0.028 \\



    





\bottomrule
\end{tabular}}
\begin{tablenotes}
\footnotesize
\item IPCW = inverse probability of censoring weighted; $X$ = baseline covariates; $Z$ = NCE; $W$ = NCO.
\end{tablenotes}
\end{table}

\vspace{-0.6cm}\section{Discussion}\label{sec:discussion}\vspace{-0.3cm}


We propose two proximal inference methods for estimating counterfactual HIV incidence under placebo using external control data with unmeasured confounding. {These methods address key methodological challenges in HIV prevention trials, including right censoring and low event rates.} Specifically, we (1) extended the semiparametric framework of \cite{cui2020semiparametric} and \cite{yingproximal} to accommodate right censoring and integrate external controls, and (2) developed a novel two-stage regression-based method tailored for low-incidence outcomes. Our approaches have broad applicability beyond HIV prevention trials. Both methods are suitable for low-incidence settings, while the IPCW method extends to moderate-to-high incidence outcomes. Table \ref{table:comparison} compares the semiparametric IPCW and two-stage regression approaches, including their assumptions, advantages, and limitations. {Semiparametric IPCW approach requires fewer modeling assumptions and does not require $Z$ to exist in the primary dataset, but the estimates and confidence intervals may not be in the support of $(0,1)$. In comparison, the two-stage regression-based approach assumes a rare-event setting, and is not easily generalizable to discrete unmeasured confounders.} Careful consideration of these methodological differences is essential for practical applications.

{We applied these methods to the HPTN 083 trial to estimate counterfactual one-year HIV cumulative incidence under placebo. Across varying analytical specifications, both the semiparametric IPCW and two-stage regression-based approaches produced consistent estimates, ranging from 4.3\% to 5.5\%. All proximal inference methods revealed significant differences between observed cabotegravir incidence and TDF/FTC incidence compared to estimated placebo incidence. These methods demonstrate the feasibility of evaluating absolute efficacy in trials without concurrent placebo groups, particularly for agents with efficacy comparable to cabotegravir.}

Identifying valid negative controls $Z$ and $W$ that serve as strong proxies for unmeasured confounders is crucial for identifying the target causal parameter in the proximal learning framework and directly impacts estimate precision. In our HIV prevention trial application, we used combined baseline rectal gonorrhea or chlamydia infection as the NCO and geographic region as the NCE, to address a major unmeasured confounder: {local HIV transmission environment risk}. In practice, multiple unmeasured confounders may be of concern, requiring proximal learning methods to identify valid NCE and NCO pairs for each confounder. Future research should apply these methods across broader HIV prevention trials to evaluate their robustness in diverse settings and assess optimal negative control selection strategies for obtaining reliable and efficient estimates of counterfactual HIV incidence.


The absolute efficacy estimated using external control methods (including our proposed methods) should be interpreted in the context of the specific placebo condition in the external dataset. For example, our placebo condition should be interpreted relative to PrEP use in the AMP study's placebo arm, which represents a clinically relevant placebo definition. Future statistical methods could be developed to estimate absolute efficacy relative to a ``pure'' placebo with no PrEP use by leveraging the varying PrEP usage patterns within the AMP placebo arm.


\begin{table}[h]
\caption{Comparison of the proposed semiparametric IPCW and two-stage regression-based approaches}
\resizebox{\textwidth}{!}{\begin{tabular}{|p{.13\textwidth}|p{.45\textwidth}|p{.45\textwidth}|}

\hline
& Semiparametric IPCW & Two-stage regression-based\\ \hline
Assumptions & \begin{enumerate}
    \item[A1.] Unmeasured confounder and valid negative controls
    \item[A2.] Positivity
    \item[A3.] Completeness
    \item[A4.] Censoring at random in the external study
\end{enumerate}  & \begin{enumerate}
    {\setlength{\itemindent}{0.2cm}\item[A$1^*$.] Continuous unmeasured confounder and valid negative controls}
    \item[A4.] Censoring at random in the external study
    \item[A5.] Low event incidence (rare events)
    \item[A6.] Unmeasured confounder mean-independent error structure 
    \item[A7.] Log-linear model for NCO and Cox model for event time 
\end{enumerate} \\ 
\hline
Advantages        & \begin{itemize}[leftmargin=1.4em]
    \item Requires fewer modeling assumptions
    \item Handles different types of unmeasured confounders depending on the type of negative control variables 
    \item NCE doesn't need to be available in the primary dataset
\end{itemize} & \begin{itemize}[leftmargin=1.4em]
    \item Typically yields more precise estimates with narrower confidence intervals in the presence of rare events
\end{itemize} \\ \hline
Limitations        & \begin{itemize}[leftmargin=1.4em]
    \item Less efficient for rare outcomes, with wider confidence intervals
    \item Estimates and confidence intervals may not be in $(0,1)$
\end{itemize} & \begin{itemize}[leftmargin=1.4em]
    \item Not easily generalizable to discrete unmeasured confounders 
    \item Needs stronger modeling assumptions
    \item Assumes a rare-event setting
\end{itemize}\\ \hline
\end{tabular}}
\label{table:comparison}
\end{table}

\clearpage

\bibliography{reference}

\clearpage
\begin{center}

{\large\bf SUPPLEMENTARY MATERIAL}

\end{center}

\setcounter{equation}{0}
\setcounter{table}{0}
\setcounter{figure}{0}
\setcounter{lemma}{1}
\setcounter{section}{0}
\renewcommand{\theequation}{S\arabic{equation}}
\renewcommand{\thelemma}{S\arabic{lemma}}
\renewcommand{\thetable}{S\arabic{table}}
\renewcommand{\thefigure}{S\arabic{figure}}
\renewcommand{\thesection}{S\arabic{section}}


\section{Additional Analysis}

\subsection{Combining all arms in AMP as the external dataset}\label{supp:allarms}

As discussed in Section \ref{amp_external_intro}, the observed HIV cumulative incidence through 1 year (\%) in AMP were 2.5 (95\% CI: [1.71, 3.53]), 2.2 (95\% CI: [1.46, 3.18]), and 2.98 (95\% CI: [2.11, 4.09]) in the low-dose, high-dose, and placebo groups, respectively. The trial did not demonstrate a statistically significant difference in HIV incidence between the VRC01 and placebo arms. Hence, as a supportive analysis, we combined all arms in AMP as the external dataset and explored rectal gonorrhea, rectal chlamydia, and either STI as the NCO. To evaluate the validity of each NCO, we fitted logistic regressions of W on Z and X. When using rectal gonorrhea, chlamydia or either STI as the NCO, we found statistically significant association between $Z$ and $W$ after adjusting for $\bmX$ (p$<$0.001, p$=$0.023, and p$<$0.001 respectively), supporting their
use as more appropriate proxies. All results are in Figure \ref{fig:all_allarms} and Table \ref{tab:sti_gnr_chlam_allarms}. 

We find that the estimated counterfactual HIV cumulative incidence through 1 year are lower across methods. This is expected since the incidence rates in the low-dose and high-dose arms in AMP are lower despite their statistical insignificance compared to the placebo arm. The confidence intervals are also smaller due to a larger sample size and a larger number of events combining all arms. However, as we mentioned, the counterfactual HIV cumulative incidence needs to be interpreted based on the ``placebo'' definition in the external data. Since in this analysis, all arms (both the placebo and the active treatment arms) with mixed PrEP use are combined as the external data, the results might be harder to interpret. 

\begin{figure}[!htbp]
    \centering
    \includegraphics[width=\linewidth]{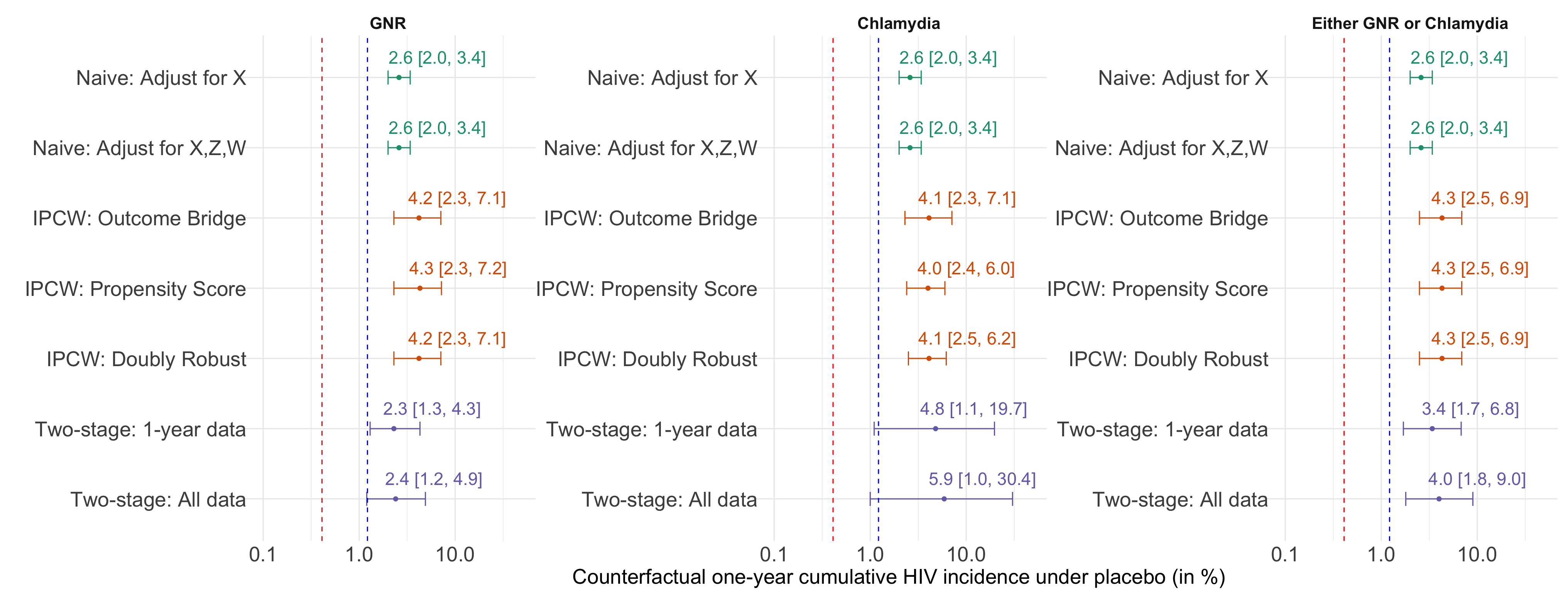}
    \caption{Estimated counterfactual HIV cumulative incidence through 1 year (\%) for HPTN 083 using different methods and NCOs. All arms in the AMP data were combined as the external control dataset. The red dashed line indicates the observed incidence rate in the Cabotegravir arm (0.41 per 100 person-years), and the blue dashed line indicates the observed incidence rate in the TDF/FTC arm (1.22 per 100 person-years). \textbf{Age, race and gender} were adjusted in all models. \textbf{All arms} in the AMP data were combined as the external data. The x-axis is plotted with the $\log()$ transformation for better visualization.}
    \label{fig:all_allarms}
\end{figure}

\begin{table}[!htbp]
\centering
\caption{Wald test statistics comparing HIV cumulative incidence through 1 year between active treatment arms and the estimated placebo arm. Statistical comparisons are presented for cabotegravir versus placebo and emtricitabine/tenofovir disoproxil fumarate (TDF/FTC) versus placebo across different analytical approaches. \textbf{Age, race and gender} were adjusted in all models. \textbf{All arms} in the AMP data were combined as the external data. }
\label{tab:sti_gnr_chlam_allarms}
\resizebox{\textwidth}{!}{
\begin{tabular}{ccccccccccc}
\toprule
{NCO} & {Method} & {Specification} & \multicolumn{4}{c}{{Cabotegravir}} & \multicolumn{4}{c}{{TDF/FTC}} \\
 &  &  & {Rel. Efficacy} & {Abs. Efficacy} & {Test Stat.} & {p-value}
 & {Rel. Efficacy} & {Abs. Efficacy} & {Test Stat.} & {p-value} \\
\midrule


& Naïve & Adjust for $X$
& 0.842 & 0.022 & -6.193 & $<0.001$
& 0.531 & 0.014 & -3.885 & $<0.001$ \\
\midrule

\multirow{6}{*}{GNR}
& Naïve & Adjust for $X,Z,W$
& 0.842 & 0.022 & -6.013 & $<0.001$
& 0.531 & 0.014 & -3.746 & $<0.001$ \\

& IPCW & Outcome Bridge
& 0.902 & 0.038 & -5.336 & $<0.001$
& 0.710 & 0.030 & -3.383 & 0.001 \\

& IPCW & Propensity Score
& 0.905 & 0.039 & -5.321 & $<0.001$
& 0.716 & 0.031 & -3.393 & 0.001 \\

& IPCW & Doubly Robust
& 0.902 & 0.038 & -5.340 & $<0.001$
& 0.710 & 0.030 & -3.388 & 0.001 \\

& Two-stage & 1-year
& 0.822 & 0.019 & -4.847 & $<0.001$
& 0.470 & 0.011 & -2.898 & 0.004 \\

& Two-stage & All data
& 0.829 & 0.020 & -4.992 & $<0.001$
& 0.492 & 0.012 & -2.957 & 0.003 \\
\midrule

\multirow{6}{*}{Chlamydia}
& Naïve & Adjust for $X,Z,W$
& 0.842 & 0.022 & -6.034 & $<0.001$
& 0.531 & 0.014 & -3.774 & $<0.001$ \\

& IPCW & Outcome Bridge
& 0.900 & 0.037 & -5.336 & $<0.001$
& 0.702 & 0.029 & -3.383 & 0.001 \\

& IPCW & Propensity Score
& 0.898 & 0.036 & -6.127 & $<0.001$
& 0.695 & 0.028 & -3.896 & $<0.001$ \\

& IPCW & Doubly Robust
& 0.900 & 0.037 & -6.179 & $<0.001$
& 0.702 & 0.029 & -3.971 & $<0.001$ \\

& Two-stage & 1-year
& 0.915 & 0.044 & -1.990 & 0.047
& 0.746 & 0.036 & -1.385 & 0.166 \\

& Two-stage & All data
& 0.931 & 0.055 & -2.939 & 0.003
& 0.793 & 0.047 & -1.801 & 0.072 \\
\midrule

\multirow{6}{*}{GNR or Chlamydia}
& Naïve & Adjust for $X,Z,W$
& 0.842 & 0.022 & -6.017 & $<0.001$
& 0.531 & 0.014 & -3.749 & $<0.001$ \\

& IPCW & Outcome Bridge
& 0.905 & 0.039 & -5.964 & $<0.001$
& 0.716 & 0.031 & -3.873 & $<0.001$ \\

& IPCW & Propensity Score
& 0.905 & 0.039 & -5.933 & $<0.001$
& 0.716 & 0.031 & -3.836 & $<0.001$ \\

& IPCW & Doubly Robust
& 0.905 & 0.039 & -5.340 & $<0.001$
& 0.716 & 0.031 & -3.873 & $<0.001$ \\

& Two-stage & 1-year
& 0.879 & 0.030 & -4.477 & $<0.001$
& 0.641 & 0.022 & -2.642 & 0.008 \\

& Two-stage & All data
& 0.898 & 0.036 & -4.639 & $<0.001$
& 0.695 & 0.028 & -2.920 & 0.004 \\

\bottomrule
\end{tabular}}
\end{table}

\subsection{Using Gonorrhea and Chlamydia as NCO}\label{supp:chlamNCO}

In this section, we present results when using rectal gonorrhea and rectal chlamydia as the NCO. Results with rectal chlamydia weren't included in the primary analysis since the logistic regression of $W\sim Z+\bmX$ suggested that it is not a good proxy of the unmeasured confounder $U$, and results with rectal gonorrhea were omitted from the primary analysis due to its low prevalence in both HPTN 083 and AMP. We present the results in Figure \ref{fig:gnr+chlam_placebo} and Table \ref{tab:gnr+chlam_placebo}. While the estimated counterfactual HIV cumulative incidence through 1 year when using rectal gonorrhea as the NCO are similar to those in the primary analysis, the confidence intervals are wider due to the low prevalence of rectal gonorrhea. 

\begin{figure}[!htbp]
    \centering
    \includegraphics[width=1\linewidth]{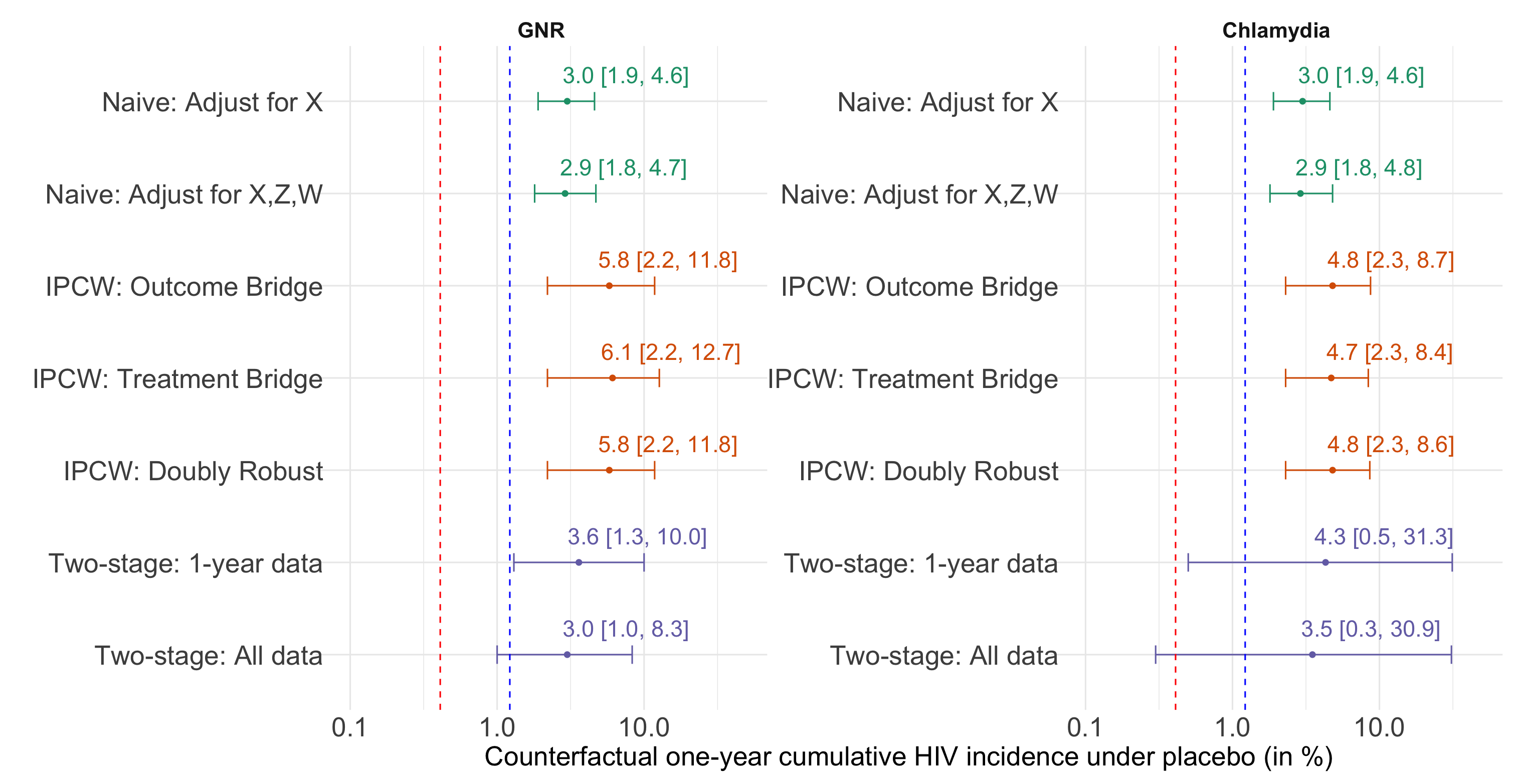}
    \caption{Estimated counterfactual HIV cumulative incidence through 1 year (\%) for HPTN 083 using different methods and rectal Gonorrhea and Chlamydia as the NCO. The red dashed line indicates the observed incidence rate in the Cabotegravir arm (0.41 per 100 person-years), and the blue dashed line indicates the observed incidence rate in the TDF/FTC arm (1.22 per 100 person-years). \textbf{Age, race and gender} were adjusted in all models. \textbf{Only the placebo arm} in the AMP data was used as the external data. The x-axis is plotted with the $\log()$ transformation for better visualization.}
    \label{fig:gnr+chlam_placebo}
\end{figure}

\begin{table}[!htbp]
\centering
\caption{Wald test statistics comparing HIV cumulative incidence through 1 year between active treatment arms and the estimated placebo arm. Statistical comparisons are presented for cabotegravir versus placebo and emtricitabine/tenofovir disoproxil fumarate (TDF/FTC) versus placebo across different analytical approaches. \textbf{Age, race and gender} were adjusted in all models. \textbf{Only the placebo arm} in the AMP data was used as the external data. }
\label{tab:gnr+chlam_placebo}
\resizebox{\textwidth}{!}{
\begin{tabular}{llccccccccc}
\toprule
{NCO} & {Method} & {Specification} & \multicolumn{4}{c}{{Cabotegravir}} & \multicolumn{4}{c}{{TDF/FTC}} \\ & & & {Rel. Efficacy} & {Abs. Efficacy} & {Test Stat.} & {p-value} & {Rel. Efficacy} & {Abs. Efficacy} & {Test Stat.} & {p-value}\\
\midrule

& Naïve & Adjust for $X$
& 0.863 & 0.026 & -5.593 & $<0.001$
& 0.593 & 0.018 & -3.292 & 0.001 \\\midrule

\multirow{6}{*}{GNR}
& Naïve & Adjust for $X,Z,W$
& 0.859 & 0.025 & -5.382 & $<0.001$
& 0.579 & 0.017 & -3.111 & 0.002 \\

& IPCW & Outcome bridge
& 0.929 & 0.054 & -4.212 & $<0.001$
& 0.790 & 0.046 & -2.859 & 0.004 \\

& IPCW & Propensity score
& 0.933 & 0.057 & -4.149 & $<0.001$
& 0.800 & 0.049 & -2.853 & 0.004 \\

& IPCW & Doubly robust
& 0.929 & 0.054 & -4.198 & $<0.001$
& 0.790 & 0.046 & -2.849 & 0.004 \\

& Two-stage & 1-year
& 0.886 & 0.032 & -2.853 & 0.004
& 0.661 & 0.024 & -1.877 & 0.061 \\

& Two-stage & All data
& 0.863 & 0.026 & -3.094 & 0.002
& 0.593 & 0.018 & -2.228 & 0.026 \\

\midrule

\multirow{6}{*}{Chlamydia}
& Naïve & Adjust for $X,Z,W$
& 0.859 & 0.025 & -5.177 & $<0.001$
& 0.579 & 0.017 & -2.900 & 0.004 \\

& IPCW & Outcome bridge
& 0.915 & 0.044 & -4.913 & $<0.001$
& 0.746 & 0.036 & -3.217 & 0.001 \\

& IPCW & Propensity score
& 0.913 & 0.043 & -4.927 & $<0.001$
& 0.740 & 0.035 & -3.202 & 0.001 \\

& IPCW & Doubly robust
& 0.915 & 0.044 & -4.923 & $<0.001$
& 0.746 & 0.036 & -3.222 & 0.001 \\

& Two-stage & 1-year
& 0.905 & 0.039 & -1.858 & 0.063
& 0.716 & 0.031 & -0.947 & 0.344 \\

& Two-stage & All data
& 0.883 & 0.031 & -1.817 & 0.069
& 0.651 & 0.023 & -0.931 & 0.352 \\

\bottomrule
\end{tabular}}
\end{table}

\subsection{Other covariates sets}\label{supp:covset}

In the primary analysis, we adjusted for age, race, gender, ethnicity, and education in models as covariates $\bmX$. In this section, we explore other adjustment sets: (1) age, race (Figure \ref{fig:agerace_placebo} and Table \ref{tab:agerace_placebo}); (2) age, race, gender, ethnicity, education (Figure \ref{fig:allcov_placebo} and Table \ref{tab:all_argee}). Note that as mentioned in Section \ref{sec:realdata}, a Cox model only showed significant associations between age and race and the outcome. We first performed regression $W\sim Z+\bmX$ to assess the strength of the NCOs. When using covariates set (1) as $\bmX$, the p-values for $Z$ is $0.143$ when using rectal chlamydia as the NCO while the p-values are less than 0.001 when using rectal gonorrhea or either STI as the NCOs. This suggests that the rectal chlamydia is a less reliable proxy of the unmeasured confounder compared to the other NCOs. Similarly, when using covariates set (2) as $\bmX$, the p-values are $0.029$, $0.138$, and $0.018$ when using rectal gonorrhea, rectal chlamydia, and either STI as the NCOs respectively. Additionally, the prevalence of rectal gonorrhea is low in both HPTN 083 and AMP, being 6.5\% and 3.7\% respectively. Hence, as in the primary analysis, using either STI combined is a better NCO. 

When adjusting for age and race, i.e. set (2), the estimation results using either STI as the NCO is similar to the primary analysis with slightly smaller confidence intervals. When adjusting for more covariates, i.e. set (2), the estimates for the counterfactual HIV cumulative incidence through 1 year are bigger with wider confidence intervals across all methods. This is possibly due to the low event rate we have and the estimation becomes unstable with 5 covariates. Across all covariate sets with either STI combined as the NCO, we see consistent statistically significant superiority of Cabotegravir over the placebo arm, while the results for TDF/FTC is mixed in the regression-based two-stage approach. 


\begin{figure}[!htbp]
    \centering
    \includegraphics[width=\linewidth]{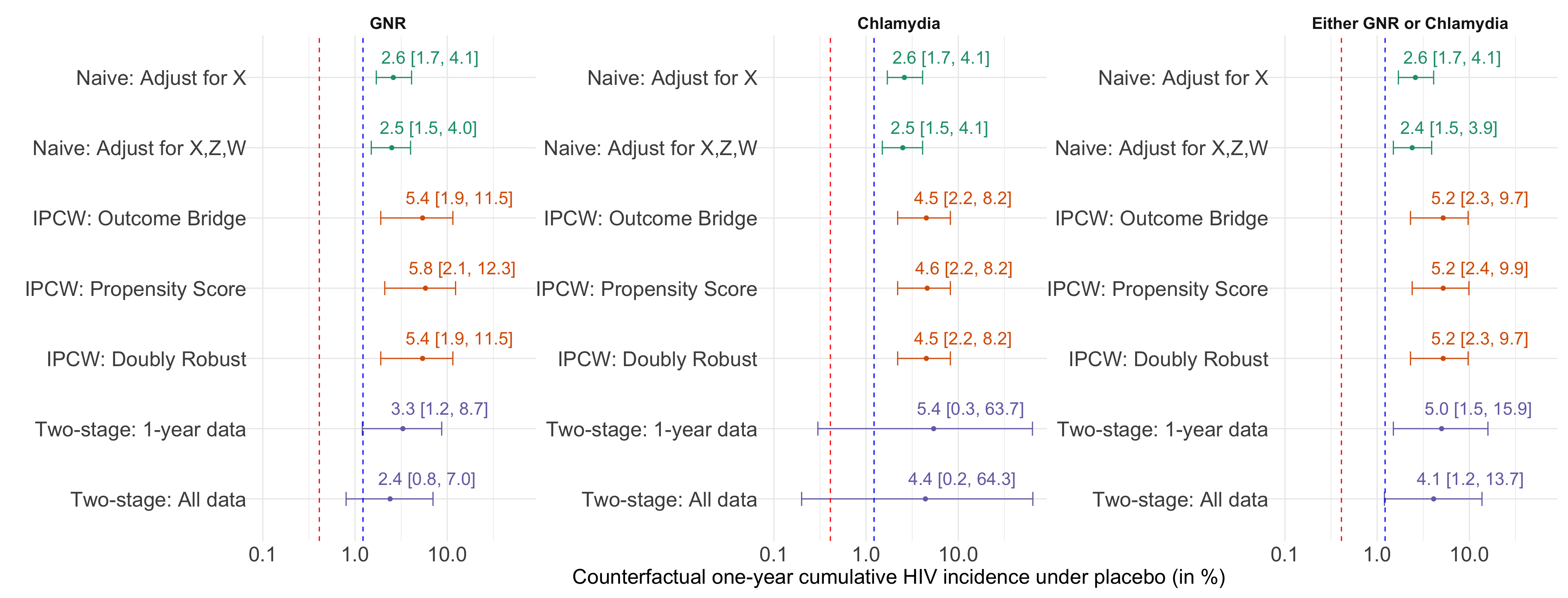}
    \caption{Estimated counterfactual HIV cumulative incidence through 1 year for HPTN 083 using different methods and definitions of the NCO. The red dashed line indicates the observed incidence rate in the Cabotegravir arm (0.41 per 100 person-years), and the blue dashed line indicates the observed incidence rate in the TDF/FTC arm (1.22 per 100 person-years). \textbf{Age and race} were adjusted in all models. \textbf{Only the placebo arm} in AMP was used as the external data. The x-axis is plotted on the $\log()$ transformation for better visualization.
    }
    \label{fig:agerace_placebo}
\end{figure}

\begin{table}[!htbp]
\centering
\caption{Wald test statistics comparing HIV cumulative incidence through 1 year between active treatment arms and the estimated placebo arm. Statistical comparisons are presented for cabotegravir versus placebo and emtricitabine/tenofovir disoproxil fumarate (TDF/FTC) versus placebo across different analytical approaches. \textbf{Age and race} were adjusted in all models. \textbf{Only the placebo arm} in the AMP data was used as the external data. }
\label{tab:agerace_placebo}
\resizebox{\textwidth}{!}{
\begin{tabular}{ccccccccccc}
\toprule
{NCO} & {Method} & {Specification} &
\multicolumn{4}{c}{{Cabotegravir}} &
\multicolumn{4}{c}{{TDF/FTC}} \\
 &  &  &
{Rel. Efficacy (\%)} & {Abs. Efficacy (\%)} & {Test Stat.} & {p-value} &
{Rel. Efficacy (\%)} & {Abs. Efficacy (\%)} & {Test Stat.} & {p-value} \\
\midrule

& Naïve & Adjust for $X$
& 0.842 & 0.022 & -5.252 & $<0.001$
& 0.531 & 0.014 & -2.842 & 0.004 \\
\midrule

\multirow{6}{*}{GNR}
& Naïve & Adjust for $X,Z,W$
& 0.836 & 0.021 & -4.874 & $<0.001$
& 0.512 & 0.013 & -2.480 & 0.013 \\

& IPCW & Outcome Bridge
& 0.924 & 0.050 & -3.922 & $<0.001$
& 0.774 & 0.042 & -2.609 & 0.009 \\

& IPCW & Propensity Score
& 0.929 & 0.054 & -4.003 & $<0.001$
& 0.790 & 0.046 & -2.715 & 0.007 \\

& IPCW & Doubly Robust
& 0.924 & 0.050 & -3.919 & $<0.001$
& 0.774 & 0.042 & -2.606 & 0.009 \\

& Two-stage & 1-year
& 0.876 & 0.029 & -2.964 & 0.003
& 0.630 & 0.021 & -1.716 & 0.086 \\

& Two-stage & All data
& 0.829 & 0.020 & -3.118 & 0.002
& 0.492 & 0.012 & -1.860 & 0.063 \\
\midrule

\multirow{6}{*}{Chlamydia}
& Naïve & Adjust for $X,Z,W$
& 0.836 & 0.021 & -4.811 & $<0.001$
& 0.512 & 0.013 & -2.442 & 0.015 \\

& IPCW & Outcome Bridge
& 0.909 & 0.041 & -4.788 & $<0.001$
& 0.729 & 0.033 & -3.077 & 0.002 \\

& IPCW & Propensity Score
& 0.911 & 0.042 & -4.864 & $<0.001$
& 0.735 & 0.034 & -3.134 & 0.002 \\

& IPCW & Doubly Robust
& 0.909 & 0.041 & -4.810 & $<0.001$
& 0.729 & 0.033 & -3.091 & 0.002 \\

& Two-stage & 1-year
& 0.924 & 0.050 & -1.605 & 0.108
& 0.774 & 0.042 & -0.888 & 0.375 \\

& Two-stage & All data
& 0.907 & 0.040 & -1.492 & 0.136
& 0.723 & 0.032 & -0.860 & 0.390 \\
\midrule

\multirow{6}{*}{GNR or Chlamydia}
& Naïve & Adjust for $X,Z,W$
& 0.829 & 0.020 & -4.780 & $<0.001$
& 0.492 & 0.012 & -2.381 & 0.017 \\

& IPCW & Outcome Bridge
& 0.921 & 0.048 & -4.703 & $<0.001$
& 0.765 & 0.040 & -3.129 & 0.002 \\

& IPCW & Propensity Score
& 0.921 & 0.048 & -4.710 & $<0.001$
& 0.765 & 0.040 & -3.146 & 0.002 \\

& IPCW & Doubly Robust
& 0.921 & 0.048 & -4.702 & $<0.001$
& 0.765 & 0.040 & -3.125 & 0.002 \\

& Two-stage & 1-year
& 0.918 & 0.046 & -3.676 & $<0.001$
& 0.756 & 0.038 & -2.116 & 0.034 \\

& Two-stage & All data
& 0.900 & 0.037 & -3.509 & $<0.001$
& 0.702 & 0.029 & -2.055 & 0.040 \\
\bottomrule
\end{tabular}}
\end{table}

\begin{figure}[!htbp]
    \centering
    \includegraphics[width=\linewidth]{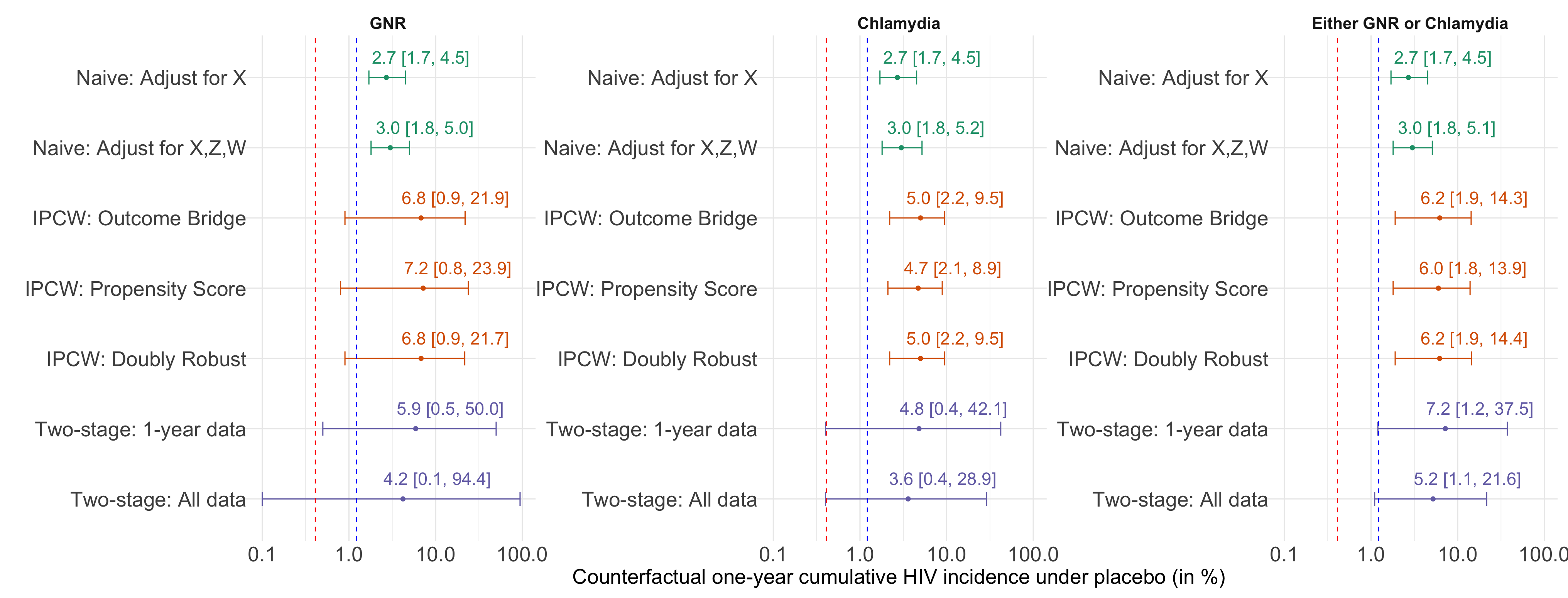}
    \caption{Estimated counterfactual HIV cumulative incidence through 1 year (\%) for HPTN 083 using different methods and definitions of the NCO. The red dashed line indicates the observed incidence rate in the Cabotegravir arm (0.41 per 100 person-years), and the blue dashed line indicates the observed incidence rate in the TDF/FTC arm (1.22 per 100 person-years). \textbf{Age, race, gender, ethnicity, and education} were adjusted in all models. \textbf{Only the placebo arm} in AMP was used as the external data. The x-axis is plotted on the $\log()$ transformation for better visualization.
    }
    \label{fig:allcov_placebo}
\end{figure}

\begin{table}[!htbp]
\centering
\caption{Wald test statistics comparing HIV cumulative incidence through 1 year between active treatment arms and the estimated placebo arm. Statistical comparisons are presented for cabotegravir versus placebo and emtricitabine/tenofovir disoproxil fumarate (TDF/FTC) versus placebo across different analytical approaches. \textbf{Age, race, gender, ethnicity, and education} were adjusted in all models. \textbf{Only the placebo arm} in the AMP data was used as the external data. }
\label{tab:all_argee}
\resizebox{\textwidth}{!}{
\begin{tabular}{ccccccccccc}
\toprule
{NCO} & {Method} & {Specification} &
\multicolumn{4}{c}{{Cabotegravir}} &
\multicolumn{4}{c}{{TDF/FTC}} \\
 &  &  &
{Rel. Efficacy (\%)} & {Abs. Efficacy (\%)} & {Test Stat.} & {p-value} &
{Rel. Efficacy (\%)} & {Abs. Efficacy (\%)} & {Test Stat.} & {p-value} \\
\midrule

& Naïve & Adjust for $X$
& 0.848 & 0.023 & -5.091 & $<0.001$
& 0.548 & 0.015 & -2.737 & 0.006 \\
\midrule

\multirow{6}{*}{GNR}
& Naïve & Adjust for $X,Z,W$
& 0.863 & 0.026 & -5.328 & $<0.001$
& 0.593 & 0.018 & -3.088 & 0.002 \\

& IPCW & Outcome Bridge
& 0.940 & 0.064 & -2.438 & 0.015
& 0.821 & 0.056 & -1.696 & 0.090 \\

& IPCW & Propensity Score
& 0.943 & 0.068 & -2.360 & 0.018
& 0.831 & 0.060 & -1.663 & 0.096 \\

& IPCW & Doubly Robust
& 0.940 & 0.064 & -2.456 & 0.014
& 0.821 & 0.056 & -1.710 & 0.087 \\

& Two-stage & 1-year
& 0.931 & 0.055 & -1.462 & 0.144
& 0.793 & 0.047 & -0.855 & 0.393 \\

& Two-stage & All data
& 0.902 & 0.038 & -2.628 & 0.009
& 0.710 & 0.030 & -1.554 & 0.120 \\
\midrule

\multirow{6}{*}{Chlamydia}
& Naïve & Adjust for $X,Z,W$
& 0.863 & 0.026 & -5.097 & $<0.001$
& 0.593 & 0.018 & -2.872 & 0.004 \\

& IPCW & Outcome Bridge
& 0.918 & 0.046 & -4.526 & $<0.001$
& 0.756 & 0.038 & -2.970 & 0.003 \\

& IPCW & Propensity Score
& 0.913 & 0.043 & -4.505 & $<0.001$
& 0.740 & 0.035 & -2.904 & 0.004 \\

& IPCW & Doubly Robust
& 0.918 & 0.046 & -4.534 & $<0.001$
& 0.756 & 0.038 & -2.976 & 0.003 \\

& Two-stage & 1-year
& 0.915 & 0.044 & -2.031 & 0.042
& 0.746 & 0.036 & -1.053 & 0.292 \\

& Two-stage & All data
& 0.886 & 0.032 & -1.942 & 0.052
& 0.661 & 0.024 & -1.028 & 0.304 \\
\midrule

\multirow{6}{*}{GNR or Chlamydia}
& Naïve & Adjust for $X,Z,W$
& 0.863 & 0.026 & -5.104 & $<0.001$
& 0.593 & 0.018 & -2.882 & 0.004 \\

& IPCW & Outcome Bridge
& 0.934 & 0.058 & -3.640 & $<0.001$
& 0.803 & 0.050 & -2.503 & 0.012 \\

& IPCW & Propensity Score
& 0.932 & 0.056 & -3.587 & $<0.001$
& 0.797 & 0.048 & -2.443 & 0.015 \\

& IPCW & Doubly Robust
& 0.934 & 0.058 & -3.606 & $<0.001$
& 0.803 & 0.050 & -2.480 & 0.013 \\

& Two-stage & 1-year
& 0.943 & 0.068 & -3.389 & 0.001
& 0.831 & 0.060 & -2.091 & 0.037 \\

& Two-stage & All data
& 0.921 & 0.048 & -3.105 & 0.002
& 0.765 & 0.040 & -1.932 & 0.053 \\
\bottomrule
\end{tabular}}
\end{table}

\section{Technical Proofs} 

\subsection{Proof of Theorem \ref{theo: censoring}}\label{proof:theo1}
\begin{proof}
    (a) We first show the identification formula using the outcome bridge function $ h $. First, we note that 
    \begin{align*}
    & E\left[\frac{\Delta I(T^*\leq t)}{P(C>T^*\mid Z, \bmX, R=1)}\mid Z, \bmX,R=1\right]\\
    = & E\left[ \frac{E[\Delta\mid T^*, Z, \bmX, R=1]I(T^*\leq t)}{P(C>T^*\mid Z, \bmX, R=1)}\mid Z, \bmX, R=1\right]\\
    = & E\left[ \frac{P(C>T^*\mid T^*, Z, \bmX, R=1)I(T^*\leq t)}{P(C>T^*\mid Z, \bmX, R=1)}\mid Z, \bmX, R=1\right]\\
     = & E\left[ \frac{P(C>T^*\mid Z, \bmX, R=1)I(T^*\leq t)}{P(C>T^*\mid Z, \bmX, R=1)}\mid Z, \bmX, R=1\right]\\
     = & E[I(T^*\leq t) \mid Z, \bmX, R=1]\\
     = & E[I(T^*(0)\leq t)\mid Z, \bmX, R=1]
\end{align*}
The first equality is by the tower property of expectation; the second equality is by definition; the third equality is by assumption \ref{assump: censoring}; the fourth equality is by algebra; and the last equality is by consistency. 
From  the equation \eqref{eq: h.c} in Theorem 1 we have 
	\begin{align*}
		E\left[\frac{\Delta I(T^*\leq t)}{P(C>T^*\mid Z, \bmX, R=1)}\mid Z, \bmX,R=1\right]=\int h(w, \bmx) dP(w\mid Z=z, R=1, \bmX=\bmx).
	\end{align*}

    Thus, 
    \begin{align*}
		E[I(T^*(0)\leq t)\mid Z=z, R=1, \bmX=\bmx] =\int h(w, \bmx) dP(w\mid Z=z, R=1, \bmX=\bmx), 
	\end{align*} holds for every $z, \bm x$. Note that left-hand side equals
	\begin{align*}
			&E[I(T^*(0)\leq t)\mid Z=z, R=1, \bmX=\bmx]  \\
			& = \int  E[I(T^*(0)\leq t)\mid U_b=u, Z=z, R=1, \bmX=\bmx] dP(u \mid Z=z, R=1, \bmX=\bmx ),
\end{align*}
and the right-hand side equals 
	\begin{align*}
		&  \int h(w, \bmx) dP(w\mid Z=z, R=1, \bmX=\bmx) \\
	&= \int \int h(w, \bmx) dP(w\mid U_b=u, Z=z, R=1, \bmX=\bmx)  dP(u \mid Z=z, R=1, \bmX=\bmx ) \\
	&= \int E[ h(W, \bmx) \mid U_b=u, Z=z, R=1, \bmX=\bmx  ] dP(u \mid Z=z, R=1, \bmX=\bmx ). 
\end{align*}
From the completeness assumption \ref{assump: completeness}(a), we have  $ E[I(T^*(0)\leq t) \mid U_b=u,Z=z, R=1, \bmX=\bmx] = E[ h(W, \bmx) \mid U_b=u, Z=z, R=1, \bmX=\bmx  ]  $ for any $ u, \bmx, z $ with positive density in $ P(u, \bmx, z\mid R=1) $.  

Then, by $ (T^*(0), W) \indep (Z, R) \mid U_b, \bmX $,  we have $ E[I(T^*(0)\leq t) \mid U_b=u, R=0, \bmX=\bmx] = E[ h(W, \bmx) \mid U_b=u,  R=0, \bmX=\bmx  ]  $ for any $ u, \bmx $ with positive density in $ P(x\mid R=1) $. From Assumption \ref{assump: positivity}, $ P( u, \bmx\mid R=0)>0 $ implies $ P(u, \bmx\mid R=1)>0 $. Thus, 
\begin{align*} 
	& P(T^*(0)\leq t) \mid R=0]\\
 &= \int E[I(T^*(0)\leq t) \mid U_b=u, R=0, \bmX=\bmx]  dP(u, \bmx\mid R=0)  \\
  &= \int E[ h(W, \bmx) \mid U_b=u, R=0, \bmX=\bmx]  dP(u, \bmx\mid R=0)  \\
   &=E[ h(W, \bmX) \mid   R=0 ] 
\end{align*}

Next, We show the identification formula using the exposure bridge function $ q $. From  $ W\indep (Z, R) \mid U_b, \bmX $ and the equation \eqref{eq: q} in Theorem 1, we have
 	\begin{align}
 \frac{P(R=0\mid W=w, \bmX=\bmx)}{P(R=1\mid W=w, \bmX=\bmx)}= \int q(z, \bmx)dP(z\mid  W=w, R=1, \bmX=\bmx) \nonumber
\end{align} holds for every $w, \bm x$.
In this equation, the left-hand side equals 
\begin{align*}
	& \frac{P(R=0\mid W=w, \bmX=\bmx)}{P(R=1\mid W=w, \bmX=\bmx)} \\
	&=  \frac{1}{P(R=1\mid W=w, \bmX=\bmx)} -1  \\
	&= \int  \frac{1}{P(R=1\mid U_b=u, W=w, \bmX=\bmx)}  d P(u\mid W=w, \bmX=\bmx, R=1)-1     \\
	&= \int  \left\{ \frac{1}{P(R=1\mid U_b=u, \bmX=\bmx)}   -1 \right\} d P(u\mid W=w, \bmX=\bmx, R=1) ,
\end{align*}
and the right-hand side equals  
\begin{align*}
 &  \int q(z, \bmx)dP(z\mid  W=w, R=1, \bmX=\bmx) \nonumber \\
 &=  \int \int q(z, \bmx) dP(z\mid U_b=u, R=1, \bmX=\bmx)  dP(u\mid  W=w,  \bmX=\bmx, R=1) \nonumber . 
\end{align*}
From the completeness assumption \ref{assump: completeness}(b), we have 
\begin{align*}
\frac{P(R=0\mid U_b, \bmX=\bmx)}{P(R=1\mid U_b, \bmX=\bmx)}   =  \int q(z, \bmx) dP(z\mid U_b,  R=1, \bmX=\bmx) 
\end{align*}
almost surely. Similarly from the above outcome bridge derivation, we can also show that 
$E[RI(T^*(0)\leq t)q(Z,X)]=E\left[\frac{Rq(Z,X)\Delta I(T^*\leq t)}{P(C>T^*\mid Z,X,R=1)}\right]$.
Thus, by $ T^*(0) \indep (R,Z)\mid (U_b, \bmX) $, 
\begin{align*}
    & \frac{1}{P(R=0)}	E\left[\frac{Rq(Z,X)\Delta I(T^*\leq t)}{P(C>T^*\mid Z,X,R=1)}\right] \\
 &= \frac{1}{P(R=0)}	E[ R  I(T^*(0)\leq t) q(Z, \bmX) ] \\
 &=   \frac{1}{P(R=0)} \int 	E[ RI(T^*(0)\leq t)  q(Z, \bmX)  \mid \bmX=\bmx] dP(\bmx) \\
 &=   \frac{1}{P(R=0)} \int 	E[  E[RI(T^*(0)\leq t)  q(Z, \bmX)\mid U_b,\bmX]  \mid \bmX=\bmx] dP(\bmx) \\
&=   \frac{1}{P(R=0)} \int 	E[  E[I(T^*(0)\leq t)\mid U_b, \bmX]  E[R q(Z, \bmX)\mid U_b,\bmX]  \mid \bmX=\bmx] dP(\bmx) \\
&=   \frac{1}{P(R=0)} \int 	E[  E[I(T^*(0)\leq t)\mid U_b, \bmX]  P(R=1 \mid U_b,\bmX)  E[q(Z, \bmX)\mid U_b,\bmX, R=1] \mid \bmX=\bmx] dP(\bmx) \\
 &=  \frac{1}{P(R=0)} \int 	E\left[  E[I(T^*(0)\leq t)  \mid U_b, \bmX]  P(R=0\mid U_b, \bmX)  \mid \bmX=\bmx \right] dP(\bmx) \\
 &=  \frac{1}{P(R=0)} \int 	E\left[  E[(1-R)I(T^*(0)\leq t)  \mid U_b, \bmX]   \mid \bmX=\bmx \right] dP(\bmx) \\
 &=  \frac{1}{P(R=0)} \int 	E\left[ (1-R) I(T^*(0)\leq t)  \mid \bmX=\bmx \right] dP (\bmx) \\  
  &=  E\left[  I(T^*(0)\leq t)  \mid R=0 \right] .
\end{align*}

(b) Then, we show the double-robust property of the identification formula \eqref{eq:doublerobust}.

First, suppose $h(W,\bmX)$ is correctly specified, but $q(Z,\bmX)$ might be wrong. The right-hand side of equation \eqref{eq:doublerobust} can be written as $$\underbrace{\frac{1}{P(R=0)}E\left[Rq(Z,\bmX)\left(\frac{\Delta I(T^*\leq t)}{P(C>T^*\mid Z,\bmX,R=1)}-h(W,\bmX)\right)\right]}_A+\underbrace{\frac{1}{P(R=0)}E\left[(1-R)h(W,\bmX)\right]}_B.$$
Notice that we have $B=E[h(W,\bmX)\mid R=0)=P(T^*(0)\leq t\mid R=0)$ by (a). Then, it remains to show $A=0$, which follows from the equation \eqref{eq: h.c} in Theorem 1. We have 
\begin{align*}
    &E\left[Rq(Z,\bmX)\left(\frac{\Delta I(T^*\leq t)}{P(C>T^*\mid Z,\bmX,R=1)}-h(W,\bmX)\right)\right]\\
    &= E\left[E\left[Rq(Z,\bmX)\left(\frac{\Delta I(T^*\leq t)}{P(C>T^*\mid Z,\bmX,R=1)}-h(W,\bmX)\right)\mid Z,X,R=1\right]\right]\\
    &=E\left[q(Z,\bmX)E\left[\left(\frac{\Delta I(T^*\leq t)}{P(C>T^*\mid Z,\bmX,R=1)}-h(W,\bmX)\right)\mid Z,X,R=1\right]\right]\\
    &=0
\end{align*}
regardless of whether $q(Z,\bmX)$ is correctly specified or not since $$E\left[\left(\frac{\Delta I(T^*\leq t)}{P(C>T^*\mid Z,\bmX,R=1)}-h(W,\bmX)\right)\mid Z,X,R=1\right]=0$$ by the equation \eqref{eq: h.c} in Theorem 1.

Then, suppose $q(Z, \bmX)$ is correct, but $h(W,\bmX)$ might be wrong. The right-hand side of the equation \eqref{eq:doublerobust} in Theorem 1 can be written as {\small $$\underbrace{\frac{1}{P(R=0)}E\left[Rq(Z,\bmX)\frac{\Delta I(T^*\leq t)}{P(C>T^*\mid Z,\bmX,R=1)}\right]}_{A'}+\underbrace{\frac{1}{P(R=0)}E\left[-Rq(Z,\bmX)h(W,\bmX)+(1-R)h(W,\bmX)\right]}_{B'}.$$}

Notice that we have $A'=P(T^*(0)\leq t\mid R=0)$ by part (a) of this proof. Then, it remains to show $B'=0$, which follows from the equation \eqref{eq: q} in Theorem 1. We have 
\begin{align*}
    &E[Rq(Z,\bmX)h(W,\bmX)]\\
    =&E\left[E\left[Rq(Z,\bmX)h(W,\bmX)\mid W, \bmX\right]\right]\\
    =& E[h(W,\bmX)E[Rq(Z,\bmX)\mid W,\bmX]]\\
    =& E[h(W,\bmX)P(R=1\mid R,\bmX)E[q(Z,\bmX)\mid W,\bmX, R=1]]\\
    =& E\left[h(W,\bmX)P(R=1\mid W,\bmX)\frac{P(R=0\mid W,\bmX)}{P(R=1\mid W, \bmX)}\right]\\
    =& E[h(W,\bmX)P(R=0\mid W, \bmX)]\\
    =& E[(1-R)h(W,\bmX)],
\end{align*}
where the fourth equality is by the equation \eqref{eq: q} in Theorem 1 and thus $B'=0$.

\end{proof}


\subsection{Proof of Theorem \ref{theo:twostage}}\label{proof:theo2}
\begin{proof}
We start with the model
\begin{align}
    \lambda(t|R,Z,X,U) = \lambda_0(t)exp(\beta_{Tx}X + \beta_{Tu} U) \nonumber
\end{align}
When the event rate is low, we have $S(t|R,U,X) \approx 1$ for all $t$.
    
Then, the density function $f(t|R,U,X) \approx \lambda(t|R,U,X)$.
    
Note that
\begin{align*}
    f(t|R,Z,X) & = E[f(t|R,U,X,Z)|R,Z,X] \\
    & = E[f(t|R,U,X)|R,Z,X] \\
    & \approx E[\lambda(t|R,U,X)|R,Z,X] \\
    & = \lambda_0(t) exp(\beta_{Tx}X) E[exp(\beta_{Tu}U)|R,Z,X]
\end{align*}

    
Assume $U = E[U|R,Z,X] + \epsilon$ where $\epsilon \indep (R,Z,X)$ and $E(\epsilon) = 0$.

Then, $f(t|R,Z,X) \approx \tilde \lambda_0(t) exp (\beta_{Tx}X + \beta_{Tu}E[U|R,Z,X])$ where $\tilde \lambda_0(t) = \lambda_0(t)E[exp(\beta_{Tu}\epsilon)]$

It remains to identify $E(U|R,Z,X)$. We can use our negative control outcome $W$. Assume
\begin{align}
    E(W|R,U,X) =  exp(\beta_{W0} + \beta_{Wx}X + \beta_{Wu}U). \nonumber
\end{align}

Note that
\begin{align*}
    E(W|R,Z,X) & = E[E(W|R,U,Z,X)|R,Z,X] \\
    & = E[E(W|R,U,X)|R,Z,X] \\
    & = E[exp(\beta_{W0} + \beta_{Wx}X + \beta_{Wu}U)|R,Z,X] \\
    & = exp(\beta_{W0}+\beta_{Wx}X)E[exp(\beta_{Wu}U)|R,Z,X]
\end{align*}

Again, using the assumption that $U = E[U|R,Z,X] + \epsilon$ where $\epsilon \indep (R,Z,X)$ and $E(\epsilon) = 0$, we have
\begin{align}
    E[W|R,Z,X] = exp(\tilde \beta_{W0} + \beta_{Wx}X + \beta_{Wu} E(U|R,Z,X)) \nonumber
\end{align}
where $\tilde \beta_{W0} = \beta_{W0} + log(E[exp(\beta_{Wu}\epsilon)])$.

Hence, $E(U|R,Z,X)$ can be expressed as a linear combination of $X$ and $log(E[W|R,Z,X])$.

This implies,
\begin{align}
    f(t|R,Z,X) \approx \tilde \lambda_0(t) exp[\beta_{Tx}^*X + \beta_{Tw}^* log(E(W|R,Z,X))] \nonumber
\end{align}
 and hence
\begin{align}
    \lambda(t|R,Z,X) \approx \tilde \lambda_0(t) exp[\tilde{\beta}_{Tx}X + \tilde{\beta}_{Tw} log(E(W|R,Z,X))]
\end{align}
where $\tilde \lambda_0(t)$, $\tilde{\beta}_{Tx}$ and $\tilde{\beta}_{Tw}$ are unknown parameters.
\end{proof}

\section{Sensitivity analysis on violation of Assumption \ref{assump: negative controls}$^*$}\label{supp:sensi_binU}

Suppose the underlying $U$ is continuous. However, in order for completeness assumption to hold when there are only binary $Z$ and $W$, we have to assume there is a binary $U_b$ which is a coarsened version of $U$. To assess how sensitive our identification is to this coarsening assumption. We develop the following bound
\begin{align*}
    \frac{1}{\Gamma_1^2} E[h(W,X)|S=0] \le E[Y(0)|S=0] \le \Gamma_1^2 E[h(W,X)|S=0],
\end{align*}
where we assume $\frac{1}{\Gamma_1} \le OR(P(R=0|u_b,x,z),P(R=0|u,z,x)\le \Gamma_1$ holds for a.e. $u,z,x$.

Assumptions needed:
\begin{itemize}
    \item Assumption 1 about negative controls: $ (Z, R)  \indep (Y(0), W) \mid \bmX, U$
    \item Assumption $3^*$ (completeness). (a) For any $ \bmx $ with positive density in the external dataset and any square-integrable function $g$, $ E[g(U_b)\mid Z, R=1, \bmX=\bmx] =0$ almost surely implies  $ g(U_b)=0 $ almost surely. 
    \item $h$ solves $E[h(W,X)|z,x,R=1] = E[Y(0)|x,z,R=1]$ for a.e. $x, z$
\end{itemize}

Note that under Assumption $3^*$, $h$ also satisfies
\begin{align*}
    E[h(W,X)|u_b, z,x,R=1] = E[Y(0)|u_b,z,x,R=1]
\end{align*}

Then,
\begin{align*}
    E[Y(0)|u_b,z,x,R=0] & = \sum_uE[Y(0)|u,u_b,z,x,R=0]P(u|u_b,z,x,R=0) \\
    & = \sum_uE[Y(0)|u,z,x,R=0]P(u|u_b,z,x,R=0) \\
    & = \sum_uE[Y(0)|u,z,x,R=1]P(u|u_b,z,x,R=0) \\
    & = \sum_uE[Y(0)|u,z,x,R=1]P(u|u_b,z,x,R=1)\frac{P(u|u_b,z,x,R=0)}{P(u|u_b,z,x,R=1)}
\end{align*}

Let $r(u,u_b,z,x) = \frac{P(u|u_b,z,x,R=0)}{P(u|u_b,z,x,R=1)}$.
\begin{align*}
    r(u,u_b,z,x) & = \frac{P(u,u_b,z,x,R=0)/P(u_b,z,x,R=0)}{P(u,u_b,z,x,R=1)/P(u_b,z,x,R=1)} \\
    & = \frac{P(R=0|u,z,x)/P(R=0|u_b,z,x)}{P(R=1|u,z,x)/P(R=1|u_b,z,x)} \\
    & = OR(P(R=0|u,z,x), P(R=0|u_b,z,x)),
\end{align*}
which by assumption is bounded between $\frac{1}{\Gamma_1}$ and $\Gamma_1$.

Hence, 
\begin{align*}
    E[Y(0)|u_b,z,x,R=0] & \le \Gamma_1 \sum_uE[Y(0)|u,z,x,R=1]P(u|u_b,z,x,R=1) \\
    & = \Gamma_1 E[Y(0)|u_b,z,x,R=1]
\end{align*}
Similarly, we have 
\begin{align*}
    E[Y(0)|u_b,z,x,R=0] & \ge \Gamma_1 \sum_uE[Y(0)|u,z,x,R=1]P(u|u_b,z,x,R=1) \\
    & = \frac{1}{\Gamma_1} E[Y(0)|u_b,z,x,R=1]
\end{align*}
Thus, 
\begin{align}
    \frac{1}{\Gamma_1} E[Y(0)|u_b,z,x,R=1]\le E[Y(0)|u_b,z,x,R=0] \le \Gamma_1 E[Y(0)|u_b,z,x,R=1] \label{sens: bd with z}
\end{align}

Next, we derive a bound on $E[Y(0)|S=0]$.

{\small \begin{align*}
    E[Y(0)|R=0] & = \sum_{u_b,z,x} E[Y(0)|u_b,z,x,R=0]P(u_b,z,x|R=0) \\
    & \le \Gamma_1 \sum_{u_b,z,x} E[Y(0)|u_b,z,x,R=1]P(u_b,z,x|R=0) \\
    & = \Gamma_1 \sum_{u_b,z,x} E[h(W,X)|u_b,z,x,R=1]P(u_b,z,x|R=0) \\
    & = \Gamma_1 \sum_{u_b,z,x}\sum_u E[h(W,X)|u,u_b,z,x,R=1]P(u|u_b,z,x,R=1)P(u_b,z,x|R=0) \\ 
    & = \Gamma_1 \sum_{u_b,z,x}\sum_u E[h(W,X)|u,z,x,R=1]P(u|u_b,z,x,R=1)P(u_b,z,x|R=0) \\ 
    & = \Gamma_1 \sum_{u_b,z,x}\sum_u E[h(W,X)|u,z,x,R=0]P(u|u_b,z,x,R=1)P(u_b,z,x|R=0) \\ 
    & = \Gamma_1 \sum_{u_b,z,x}\sum_u E[h(W,X)|u,z,x,R=0]P(u|u_b,z,x,R=0)P(u_b,z,x|R=0)\frac{P(u|u_b,z,x,R=1)}{P(u|u_b,z,x,R=0)} \\ 
    & = \Gamma_1 \sum_{u_b,z,x}\sum_u E[h(W,X)|u,z,x,R=0]P(u,z,x|R=0)\frac{P(u|u_b,z,x,R=1)}{P(u|u_b,z,x,R=0)} \\ 
\end{align*}}
Note that $\frac{P(u|u_b,z,x,R=1)}{P(u|u_b,z,x,R=0)} = \frac{1}{r(u,u_b,z,x)}$ is also bounded by $\frac{1}{\Gamma_1}$ and $\Gamma_1$.

Hence, 
\begin{align*}
    E[Y(0)|R=0] & \le \Gamma_1^2 \sum_{u_b,z,x}\sum_u E[h(W,X)|u,z,x,R=0]P(u,z,x|R=0) \\
    & = \Gamma_1^2 E[h(W,X)|R=0]
\end{align*}

Similarly, we have
\begin{align*}
    E[Y(0)|R=0] & \ge \Gamma_1^2 \sum_{u_b,z,x}\sum_u E[h(W,X)|u,z,x,R=0]P(u,z,x,R=0) \\
    & = \frac{1}{\Gamma_1^2} E[h(W,X)|R=0]
\end{align*}
and therefore,
\begin{align*}
    \frac{1}{\Gamma_1^2} E[h(W,X)|R=0]\le E[Y(0)|R=0] \le \Gamma_1^2 E[h(W,X)|R=0]
\end{align*}

We now develop a similar result based on identification using treatment bridge function $q$. Specifically, 
\begin{align*}
    \frac{1}{\Gamma_2^2} \frac{1}{P(R=0)}E[I(R=1)q(Z,X)Y] \le E[Y(0)|S=0] \le \Gamma_2^2 \frac{1}{P(R=0)}E[I(R=1)q(Z,X)Y],
\end{align*}
where we assume $\frac{1}{\Gamma_2} \le OR(P(R=0|u_b,w,x),P(R=0|u,w,x)\le \Gamma_2$ and 
\begin{align*}
 \frac{1}{\Gamma_2} \le \frac{E[Y(0)|u_b,w,x,R=1]}{E[Y(0)|u,u_b,w,x,R=1]} \le \Gamma_2
\end{align*}
holds for a.e. $u,w,x$.

Assumptions needed:
\begin{itemize}
    \item Assumption 1 about negative controls: $ (Z, R)  \indep (Y(0), W) \mid \bmX, U$
    \item Assumption $3^*$ (completeness). (a) For any $ \bmx $ with positive density in the external dataset and any square-integrable function $g$, $ E[g(U_b)\mid W, R=1, \bmX=\bmx] =0$ almost surely implies  $ g(U_b)=0 $ almost surely. 
    \item $q$ solves $\frac{P(R=0|w,x)}{P(R=1|w,x)} = E[q(Z,X)|w,x,R=1] ]$ for a.e. $x, w$
\end{itemize}

By assumption, $q$ solves 
\begin{align}
    \frac{P(R=0|W,X)}{P(R=1|W,X)} = E[q(Z,X)|W,X,R=1] ]\label{equation for q}
\end{align}
which is equivalent to
\begin{align*}
    P(R=0|W,X) = E[I(R=1)q(Z,X)|W,X]
\end{align*}

Under the assumption $3*$, $q$ also satisfies
\begin{align*}
    P(R=0|U_b,W,X) = E[I(R=1)q(Z,X)|U_b,W,X]
\end{align*}

The reason is as follows. The LHS of \ref{equation for q}
\begin{align*}
    \frac{P(R=0|W,X)}{P(R=1|W,X)} & = 1-\frac{1}{P(R=1|W,X)} \\
    & = 1 - \sum_{u_b} \frac{1}{P(R=1|u_b,W,X)} P(u_b|W,X,R=1) \\
    & = \sum_{u_b} \frac{P(R=0|u_b,W,X)}{P(R=1|u_b,W,X)} P(u_b|W,X,R=1) \\
    & = E[\frac{P(R=0|u_b,W,X)}{P(R=1|u_b,W,X)}|W,X,R=1]
\end{align*}
The RHS of \ref{equation for q} is
\begin{align*}
    E[q(Z,X)|W,X,R=1] = E[E[q(Z,X)|U_b,W,X,R=1]|W,X,R=1]
\end{align*}
By the completeness assumption based on $U_b$, we conclude 
\begin{align*}
    \frac{P(R=0|U_b,W,X)}{P(R=1|U_b,W,X)} = E[q(Z,X)|U_b,W,X,R=1]
\end{align*}
which implies the desired result.

By similar derivations as in the proof of \eqref{sens: bd with z}, we have
\begin{align}
    \frac{1}{\Gamma_2} E[Y(0)|u_b,w,x,R=1]\le E[Y(0)|u_b,w,x,R=0] \le \Gamma_2 E[Y(0)|u_b,w,x,R=1]. \label{sens: bd with w}
\end{align}

Then,
{\small \begin{align*}
    E[Y(0)|R=0] & = \frac{1}{P(R=0)} \sum_{u_b,w,x}E[Y(0)|u_b,w,x,R=0]P(u_b,w,x,R=0) \\
    & \le \frac{\Gamma_2}{P(R=0)} \sum_{u_b,w,x}E[Y(0)|u_b,w,x,R=1]P(R=0|u_b,w,x)P(u_b,w,x) \\
    & = \frac{\Gamma_2}{P(R=0)} \sum_{u_b,w,x} E[Y(0)|u_b,w,x,R=1] E[I(R=1)q(Z,X)|u_b,w,x] P(u_b,w,x) \\
    & = \frac{\Gamma_2}{P(R=0)} \sum_{u_b,w,x} E[Y(0)|u_b,w,x,R=1] E[q(Z,X)|u_b,w,x,R=1] P(R=1|u_b,w,x) P(u_b,w,x) \\
    & = \frac{\Gamma_2}{P(R=0)} \sum_{u_b,w,x} E[Y(0)|u_b,w,x,R=1] \bigg\{ \sum_u E[Y(0)q(Z,X)|u,u_b,w,x,R=1] \\
    & P(u|u_b,w,x,R=1)\bigg \}P(R=1|u_b,w,x) P(u_b,w,x) \\
    & = \frac{\Gamma_2}{P(R=0)} \sum_{u_b,w,x} \bigg\{ \sum_u \frac{E[Y(0)|u_b,w,x,R=1]}{E[Y(0)|u,u_b,w,x,R=1]} E[Y(0)q(Z,X)|u,u_b,w,x,R=1]\\
    & P(u|u_b,w,x,R=1)\bigg \} P(R=1|u_b,w,x) P(u_b,w,x)
\end{align*}}

By the assumption that, for any $u, u_b, w, x$,
\begin{align*}
 \frac{1}{\Gamma_2} \le \frac{E[Y(0)|u_b,w,x,R=1]}{E[Y(0)|u,u_b,w,x,R=1]} \le \Gamma_2
\end{align*}
Then,
\begin{align*}
    E[Y(0)|R=0] & = \frac{\Gamma_2}{P(R=0)} \sum_{u_b,w,x} \bigg\{ \sum_u \frac{E[Y(0)|u_b,w,x,R=1]}{E[Y(0)|u,u_b,w,x,R=1]} E[Y(0)q(Z,X)|u,u_b,w,x,R=1]\\
    & P(u|u_b,w,x,R=1)\bigg \} P(R=1|u_b,w,x) P(u_b,w,x) \\
    & \le \frac{\Gamma_2}{P(R=0)} \sum_{u_b,w,x} \bigg\{ \sum_u E[Y(0)q(Z,X)|u,u_b,w,x,R=1]\\
    & P(u|u_b,w,x,R=1)\bigg \} P(R=1|u_b,w,x) P(u_b,w,x)\\
    & = \frac{\Gamma_2^2}{P(R=0)} \sum_{u_b,w,x} E[Y(0)q(Z,X)|u_b,w,x,R=1] P(R=1|u_b,w,x) P(u_b,w,x) \\
    & = \frac{\Gamma_2^2}{P(R=0)} \sum_{u_b,w,x} E[I(R=1)Y(0)q(Z,X)|u_b,w,x] P(u_b,w,x) \\
    & = \frac{\Gamma_2^2}{P(R=0)}  E[I(R=1)Yq(Z,X)] 
\end{align*}

Similarly, we can show that
\begin{align*}
    \frac{1}{\Gamma_2^2} \frac{1}{P(R=0)}E[I(R=1)q(Z,X)Y] \le E[Y(0)|S=0].
\end{align*}

In the above derivations, we highlight the following two assumptions:

1. Let $r(u,u_b,w,x) = \frac{P(u|u_b,w,x,R=0)}{P(u|u_b,w,x,R=1)}$.
\begin{align*}
    r(u,u_b,w,x) & = \frac{P(u,u_b,w,x,R=0)/P(u_b,w,x,R=0)}{P(u,u_b,w,x,R=1)/P(u_b,w,x,R=1)} \\
    & = \frac{P(R=0|u,w,x)/P(R=0|u_b,w,x)}{P(R=1|u,w,x)/P(R=1|u_b,w,x)} \\
    & = OR(P(R=0|u,w,x), P(R=0|u_b,w,x)),
\end{align*}
which is assumed to be bounded between $\frac{1}{\Gamma_2}$ and $\Gamma_2$.

2. Let $\tilde r(u,u_b,w,x) = \frac{E[q(Z,X)|u_b,w,x,R=1]}{E[q(Z,X)|u,u_b,w,x,R=1]}$. We assume 
\begin{align*}
 \frac{1}{\Gamma_2} \le \tilde r(u,u_b,w,x) \le \Gamma_2.
\end{align*}

Note that by the completeness assumption 3*,
\begin{align*}
    E[q(Z,X)|u_b,w,x,R=1] = \frac{P(R=0|u_b,w,x)}{P(R=1|u_b,w,x)}
\end{align*}
If it also holds that 
\begin{align*}
    E[q(Z,X)|u,w,x,R=1] = \frac{P(R=0|u,w,x)}{P(R=1|u,w,x)}
\end{align*}
Then, the above two assumptions can be unified.

Finally, for doubly robust identification, let
\begin{align*}
    \psi = \frac{1}{P(R=0)}E[I(R=1)q(Z,X)(Y - h(W,X))] + \frac{1}{P(R=0)}E[I(R=0)h(W,X)].
\end{align*}
We know:

(1)  if $h(W,X)$ satisfies $E[Y|Z,X,R=1] = E[h(W,X)|Z,X,R=1]$, then
\begin{align*}
    \psi = \frac{1}{P(R=0)}E[I(R=0)h(W,X)]
\end{align*}
(2) if $q(Z,X)$ satisfies $P(R=0|W,X) = E[I(R=1)q(Z,X)|W,X]$, then
\begin{align*}
    \psi = \frac{1}{P(R=0)}E[I(R=1)q(Z,X)Y]
\end{align*}
In the first case,  we have shown
\begin{align*}
   \frac{1}{\Gamma_1^2}\psi \le E[Y(0)|R=0] \le \Gamma_1^2  \psi
\end{align*}
In the second case, we have shown
\begin{align*}
   \frac{1}{\Gamma_2^2}\psi \le E[Y(0)|R=0] \le \Gamma_2^2  \psi
\end{align*}
Therefore, we have
\begin{align*}
   \frac{1}{\min\{\Gamma_1,\Gamma_2\}^2}\psi \le E[Y(0)|R=0] \le \min\{\Gamma_1,\Gamma_2\}^2  \psi
\end{align*}

We conduct a simple simulation study to investigate the consequence of the violation of Assumption \ref{assump: negative controls}$^*$. The data are generated as follows:
\begin{itemize}
    \item $X \sim {\rm Binom}(n, p=0.3)$
    \item $U \sim {\rm Truncnorm}(n,{\rm mean=}0.5, {\rm sd=}1)$
    \item $U_b={\rm I}(U>a)$
    \item $R \sim {\rm Binom}(n,p=1/(1+{\rm exp}(-1.6+5U+0.1X)))$
    \item $W \sim {\rm Binom} (n, p={\rm exp}(-2+2U-0.2X))$
    \item $Z \sim {\rm Binom}(n,p=1/(1+{\rm exp}(1.5U_b+0.5X)))$
\end{itemize}

The mean survival time $T^*$ is 600 months and 4000 months for $R=1$ and $R=0$. The hazard ratios for the true survival time are simulated by $\exp(0.8X+2U)$ and $\exp(0.5X-1.8U)$. The mean censoring time $C$ is 30 months and 50 months for $R=1$ and $R=0$. The hazard ratios for the censoring time are simulated by $\exp(0.2Z+0.2X)$ and $\exp(0.2Z+0.3X)$. The observed time is $T=\min (T^*, C)$. We consider sample sizes being $n=6500$, and vary $a = 0.5, 0.9$ (i.e., $U_b$ being a good or weak proxy of $U$). The table below summarizes the results.

\begin{table}[ht]
\centering
\caption{Comparison of IPCW estimates across varying $a$ in the violation of Assumption \ref{assump: negative controls}$^*$. The total number of simulations is 1000.}\label{tbl: vio_assump1}
\begin{tabular}{ccccc}
\hline
  & \multicolumn{2}{c}{$a=0.5$} & \multicolumn{2}{c}{$a=0.9$} \\
\cline{2-5}
 & Bias & SE & Bias & SE \\
\hline
    Outcome Bridge & 0.013 & 0.236 & 0.024 & 0.411 \\
    Propensity Bridge & 0.010 & 0.234 & 0.021 & 0.411 \\
    Doubly Robust & 0.013 & 0.236 & 0.024 & 0.411 \\
\hline
\end{tabular}
\end{table}

When $a=0.5$, the correlations between $U_b$ and $U$ as well as $U_b$ and $Z$ are stronger since $U_b$ is less deterministic. Hence, when the true data generating process of $R$ depends on continuous $U$ (a violation of Assumption \ref{assump: negative controls}$^*$), we observe smaller bias and SE with a stronger proxy variable $Z$.

\section{When $W$ or $Z$ is not a proxy of $U$}\label{supp:sensi_proxyWZ}

In this section, we aim to investigate the scenario where $Z$ may not be a proxy of $U$ and how this affects the identification result presented in Theorem 1. Specifically, we no longer assume $Z$ satisfies Assumption 3(a) and 4. Instead, we only assume $Z$ satisfies 
\begin{align*}
    Z \indep (U,Y(0),W)|\bm X,R=1.
\end{align*}
This assumption holds if $U$ does not have a direct effect on $Z$.

Under this assumption, at the population level, there no longer exists a square-integrable function $q$ that solves the following observed integral equation
\begin{align}
    \frac{P(R=0|W,X)}{P(R=1|W,X)} = E[q(Z,X)|W,X,R=1]. \label{int: q obs}
\end{align}
because the right hand side equals $E[q(Z,X)|X,R=1]$ which does not vary with $W$. Therefore, the identification formula (4) in Theorem 1 is invalid. In fact, the identification formula (3) in Theorem 1 is likely invalid either, because without assumption 3(a),  the outcome bridge function $h$ that solves the following observed outcome bridge equation may not satisfy the latent outcome bridge equation $E[Y(0)|U,X,R=1] = E[h(W,X)|U,X,R=1]$, a key result for controlling for unmeasured confounding. 
\begin{align}
    E[Y(0)|Z,X,R=1] = E[h(W,X)|Z,X,R=1]. \label{int: h obs}
\end{align}
This indicates that when $Z$ is not a proxy of $U$, $E[Y(0)|R=0]$ may be even unidentifiable.

However, in practice, one may not know whether or not $Z$ is a good enough proxy of $U$ such that Assumption 3(a) and 4 hold, and may proceed to use the identification formula (4) in Theorem 1 for analysis. That is, one may still estimate $q$ by solving the observed integral equation \ref{int: q obs} and then plug it into the identification formula (4) to obtain an estimate of $E[Y(0)|R=0]$. Note that although the observed integral equation in \ref{int: q obs} does not have a solution at the population level, it could still admit a solution in finite-sample, e.g., when $W$, $Z$ and $X$ are categorical and $q$ is a saturated model. What would the result based on the identification formula (4) converge to? We perform the following finite-sample analysis, assuming all variables are categorical.
{\small \begin{align*}
    E_n[I(R=1)\hat q(Z,X)Y]/\hat P(R=0) & = E_n[E_n[I(R=1)Y\hat q(Z,X)|W,X]]/\hat P(R=0) \\
    & = E_n[E_n[Y\hat q(Z,X)|W,X,R=1]\hat P(R=1|W,X)]/\hat P(R=0) \\
    & = E_n[E_n[Y\hat q(Z,X)|W,X,R=1]\hat P(R=1|W,X)]/\hat P(R=0) \\
    & \approx E_n[E_n[Y|W,X,R=1]E_n[q(Z,X)|W,X,R=1]\hat P(R=1|W,X)]/\hat P(R=0) \\
    & = E_n[E_n[Y|W,X,R=1]E_n[q(Z,X)|W,X,R=1]\hat P(R=1|W,X)]/\hat P(R=0) \\
    & = E_n[E_n[Y|W,X,R=1]\hat P(R=0|W,X)]/\hat P(R=0) \\
    & = E_n[I(R=0)E_n[Y|W,X,R=1]]/\hat P(R=0)
\end{align*}}
where the approximation on the fourth line is due to independence of $Z$ with $Y$ given $X$ may not hold exactly in a finite sample and the fifth equality is because we assume $q$ solves the observed integral equation in a finite sample exactly. This non-rigorous analysis indicates that $E_n[I(R=1)\hat q(Z,X)Y]/\hat P(R=0)$ would converge to $E[I(R=0)E[Y|W,X,R=1]]/P(R=0)$. It is easy to show that $E[I(R=0)E[Y|W,X,R=1]]/P(R=0) \neq E[Y(0)|R=0]$, unless $E[Y(0)|W,X,R=1] = E[Y(0)|W,X,R=0]$ almost surely. 

In finite samples, based on the above derivation, we argue that on average,  $E_n[I(R=1)\hat q(Z,X)Y]/\hat P(R=0)$ would approximate $E_n[I(R=0)E_n[Y|W,X,R=1]]/\hat P(R=0)$, as long as $\hat q$ solves the following finite-sample integral equation exactly:
\begin{align*}
    \frac{\hat P(R=0|W,X)}{\hat P(R=1|W,X)} = E_n[q(Z,X)|W,X,R=1].
\end{align*}
$E_n[I(R=0)E_n[Y|W,X,R=1]]/\hat P(R=0)$ is not consistent for $E[Y(0)|R=0]$ with the magnitude of bias depending on how close $E[Y(0)|W,X,R=1]$ is to $E[Y(0)|W,X,R=0]$, or in other words how strong $W$ is as a proxy of $U$. Furthermore, because $q$ is actually not identified, we expect $\hat q$ to only capture noise and exhibit large instability, which renders the plug-in estimator $E_n[I(R=1)\hat q(Z,X)Y]/\hat P(R=0)$ unstable.

To further confirm these observations, we conduct a simulation study. The data are generated as follows:
\begin{itemize}
    \item $X \sim {\rm Binom}(n,0.5)$
    \item $U \sim {\rm Binom}(n,0.5)$
    \item $R \sim {\rm Binom}(n,1/(1+{\rm exp}(-0.3+1.5U+0.5X)))$
    \item $W \sim {\rm Binom} (n, {\rm exp}(-w_u+w_uU-0.2X))$
    \item $Z \sim {\rm Binom}(n,1/(1+{\rm exp}(-0.5+0.5X)))$
\end{itemize}

We simulate the true survival using a similar approach as in the primary simulation study, but we simplify it by assuming there is no censoring. The mean survival time $T^*$ is 600 months and 4000 months for $R=1$ and $R= 0$. The hazard ratios for the true survival time are simulated by $exp(-1+0.8X+2U)$ and $exp(-1+0.5X+U)$. We obtain the observed 1-year HIV-free outcome by $Y=I(T^*>12)$. We consider sample sizes being 1e4, 1e5, 1e6, and vary $w_u = 0.5, 2, 8$ (i.e., $W$ being a  weak, moderate and strong proxy of $U$). In this setting, the target parameter $E[Y(0)|R=0] = 0.959$. The table below summarizes the results.

\begin{table}[!htbp]
\centering
\caption{Comparison of two estimators across varying $\lambda$ and sample sizes $n$. SD is calculated by the $97.5$-th quantile minus $2.5$-th quantile divided by $2\times 1.96$. 5000 simulations.}\label{tbl: z not proxy}
\begin{tabular}{cc|cc|cc}
\hline
 &  & \multicolumn{2}{c}{$E_n[I(R=1)\hat q(Z,X)Y]/\hat P(R=0)$} & \multicolumn{2}{c}{$E_n[I(R=0)E_n[Y|W,X,R=1]]/\hat P(R=0)$} \\
\cline{3-6}
$w_u$ & $N$ & Median & SD & Median & SD \\
\hline
0.5 & $10^4$ & 0.977 & 0.249 & 0.976 & 0.003 \\
    & $10^5$ & 0.977 & 0.235 & 0.976 & 0.001 \\
    & $10^6$ & 0.977 & 0.284 & 0.976 & 0.001 \\
\hline
2   & $10^4$ & 0.972 & 0.543 & 0.969 & 0.004 \\
    & $10^5$ & 0.967 & 0.535 & 0.968 & 0.002 \\
    & $10^6$ & 0.970 & 0.556 & 0.968 & 0.001 \\
\hline
8   & $10^4$ & 0.964 & 0.627 & 0.963 & 0.005 \\
    & $10^5$ & 0.961 & 0.630 & 0.963 & 0.002 \\
    & $10^6$ & 0.962 & 0.726 & 0.963 & 0.002 \\
\hline
\end{tabular}

\end{table}

Observations: (1) the bias in the estimator $E_n[I(R=0)E_n[Y|W,X,R=1]]/\hat P(R=0)$ decreases with $w_u$, which agrees with our analysis. (2) The estimator $E_n[I(R=1)\hat q(Z,X)Y]/\hat P(R=0)$ provides estimates on average close to that from $E_n[I(R=0)E_n[Y|W,X,R=1]]$, but the former has much larger variance due to instability of $\hat q$.

Next, we perform a similar analysis, assuming $W$ is not a proxy of $U$. In this analysis, we assume $W$ satisfies
\begin{align*}
    W \indep (U,R,Z)|X.
\end{align*}
When this assumption holds, at the population level, there no longer exists a square integrable function $h$ that solves the observed integral equation
\begin{align}
    E[Y|Z,X,R=1] = E[h(W,X)|Z,X,R=1], \label{int: h obs}
\end{align}
because the right hand side equals $E[h(W,X)|X,R=1]$ which does not vary with $Z$. Therefore, the identification formula (3) in Theorem 1 in invalid. If one proceed with the plug-in estimator based on this identification formula, and solve $\hat h$ using the above integral equation in finite samples. Then, we can show that
\begin{align*}
    E_n[\hat h(W,X)|R=0] & = E_n[E_n[\hat h(W,X)|X,R=0]|R=0] \\
    & \approx E_n[E_n[\hat h(W,X)|X,R=1]|R=0] \\
    & = E_n[E_n[E_n[\hat h(W,X)|Z,X,R=1]|X,R=1]|R=0] \\
    & = E_n[E_n[E_n[Y(0)|Z,X,R=1]|X,R=1]|R=0] \\
    & = E_n[E_n[Y(0)|X,R=1]|R=0]
\end{align*}
which does not equal to $E[Y(0)|R=0]$ unless $E[Y(1)|X,R=1]=  E[Y(1)|X,R=0]$.
Similar to our comments above, this derivation is in fact invalid because there does not exist $h$ that solves \ref{int: h obs} at the population level. However, this analysis suggests, in finite sample, if there exists $\hat h$ that solves $\ref{int: h obs}$ exactly in finite sample, then on avearge, we expect $E_n[\hat h(W,X)|R=0]$ equals $E_n[E_n[Y|X,R=1]|R=0]$, and $E_n[E_n[Y|X,R=1]|R=0]$ is biased for $E[Y(0)|R=0]$ with the magnitude of bias depending on how close $E[Y(1)|X,R=1]$ is to $E[Y(1)|X,R=0]$.

To confirm this, we conduct a simulation study. The data are generated as follows:
\begin{itemize}
    \item $X \sim {\rm Binom}(n,0.5)$
    \item $U \sim {\rm Binom}(n,0.5)$
    \item $R \sim {\rm Binom}(n,1/(1+{\rm exp}(-0.3+1.5U+0.5X)))$
    \item $W \sim {\rm Binom} (n, {\rm exp}(-0.5-0.5X))$
    \item $Z \sim {\rm Binom}(n,1/(1+{\rm exp}(-z_u/2+z_uU-0.2X)))$
\end{itemize}

The outcome is simulated similarly as in Table \ref{tbl: z not proxy}. We consider sample sizes being 1e4, 1e5, 1e6, and vary $z_u = 0.5, 2, 8$ (i.e., $Z$ being a  weak, moderate and strong proxy of $U$). In this setting, the target parameter $E[Y(0)|R=0] = 0.959$. The table below summarizes the results.

\begin{table}[!htbp]\label{tbl: w not proxy}
\centering
\caption{Comparison of two estimators across varying $\lambda$ and sample sizes $n$. SD is calculated by the $97.5$-th quantile minus $2.5$-th quantile divided by $2\times 1.96$. 5000 simulations.}
\begin{tabular}{cc|cc|cc}
\hline
 &  & \multicolumn{2}{c}{$E_n[I(R=1)\hat q(Z,X)Y]/\hat P(R=0)$} & \multicolumn{2}{c}{$E_n[I(R=0)E_n[Y|W,X,R=1]]/\hat P(R=0)$} \\
\cline{3-6}
$w_u$ & $N$ & Median & SD & Median & SD \\
\hline
0.5 & $10^4$ & 0.977 & 0.249 & 0.976 & 0.003 \\
    & $10^5$ & 0.977 & 0.235 & 0.976 & 0.001 \\
    & $10^6$ & 0.977 & 0.284 & 0.976 & 0.001 \\
\hline
2   & $10^4$ & 0.972 & 0.543 & 0.969 & 0.004 \\
    & $10^5$ & 0.967 & 0.535 & 0.968 & 0.002 \\
    & $10^6$ & 0.970 & 0.556 & 0.968 & 0.001 \\
\hline
8   & $10^4$ & 0.964 & 0.627 & 0.963 & 0.005 \\
    & $10^5$ & 0.961 & 0.630 & 0.963 & 0.002 \\
    & $10^6$ & 0.962 & 0.726 & 0.963 & 0.002 \\
\hline
\end{tabular}

\end{table}

Observations: (1) the bias in the estimator $E_n[E_n[Y(0)|X,R=1]|R=0]$ does not change with $z_u$, which agrees with our analysis. (2) The estimator $E_n[\hat h(W,X)|R=0]$ provides estimates on average close to that from $E_n[E_n[Y(0)|X,R=1]|R=0]$, but the former has much larger variance due to instability of $\hat h$.

These analyses underscore the importance of assessing the strength of proxy associations in practice. Notably, unlike in instrumental variable analysis—where estimators tend to be unstable when instruments are weak—proximal causal inference may remain stable with weak proxies, particularly in small samples. However, without evaluating the strength of proxies, users who apply proximal methods may be misled by seemingly precise estimates that are, in fact, biased and increasingly unstable as sample size grows.

\section{When $h$ and $q$ solve observed integral equations exactly}\label{supp:equiv_hq}

When $W,Z,X$ are all categorical variables, like in our application, it is possibly that $\hat h$ and $\hat q$ solve the observed integral equation exactly, especially when $h$ and $q$ are estimated by some saturated models. That is, we have 
\begin{align*}
    \frac{\hat P(R=0|W,X)}{\hat P(R=1|W,X)} = E_n[q(Z,X)|W,X,R=1]
\end{align*}
and
\begin{align*}
    E_n[Y|Z,X,R=1] = E_n[\hat h(W,X)|Z,X,R=1].
\end{align*}
In this case, our observation is that the plug-in estimators using the identification formula in (3) and (4) of Theorem 1 coincide. This can be justified by the following derivations.
\begin{align*}
    E[I(R=1)Yq(Z,X)] & = E[q(Z,X)E[I(R=1)Y|Z,X]] \\
    & = E[q(Z,X)E[Y|Z,X,R=1]P(R=1|Z,X)] \\
    & = E[q(Z,X,)P(R=1|Z,X)E[h(W,X)|Z,X,R=1]] \\
    & = E[q(Z,X)E[I(R=1)h(W,X)|Z,X]] \\
    & = E[I(R=1)h(W,X)q(Z,X)] \\
    & = E[h(W,X)E[I(R=1)q(Z,X)|W,X]] \\
    & = E[h(W,X)P(R=1|W,X)E[q(Z,X)|W,X,R=1]]\\
    & = E[h(W,X)P(R=0|W,X)] \\
    & = E[h(W,X)I(R=0)]
\end{align*}
This derivation holds in finite samples by replacing population expectation $E$ with sample expectation.

\section{Parametric Two-stage Regression with Binary $U$}\label{supp:twostage_binU}

We start with the model
    \begin{align}
        \lambda(t|R,Z,X,U) = \lambda_0(t)exp(\beta_{Tx}X + \beta_{Tu} U) \nonumber
    \end{align}
    where $U \in \{0, 1\}$ is binary.
    
    When the event rate is low, we have $S(t|R,U,X) \approx 1$.
    
    Then, the density function $f(t|R,U,X) \approx \lambda(t|R,U,X)$
    
    Note that
    \begin{align*}
        f(t|R,Z,X) & = E[f(t|R,U,X,Z)|R,Z,X] \\
        & = E[f(t|U,X)|R,Z,X] \\
        & \approx E[\lambda(t|U,X)|R,Z,X] \\
        & = \lambda_0(t) exp(\beta_{Tx}X) E[exp(\beta_{Tu}U)|R,Z,X]
    \end{align*}

    Since $U$ is binary, we know
    \begin{align*}
        E[exp(\beta_{T_u} U) | R,Z,X] & = exp(\beta_{T_u})P(U=1|R,Z,X) + P(U=0|R,Z,X) \\
        & = (exp(\beta_{T_u}) - 1)P(U=1|R,Z,X) + 1
    \end{align*}

    It remains to identify $P(U=1|R,Z,X)$. Assume a log-linear model for binary $W$:
    \begin{align}
        P(W=1|R,Z,U,X) =  exp(\beta_{W0} + \beta_{Wx}X + \beta_{Wu}U) \nonumber
    \end{align}

    Then,
    {\small \begin{align*}
        P(W=1|R,Z,X) & = P(W=1|R,Z,X,U=1)P(U=1|R,Z,X) + P(W=1|R,Z,X,U=0)P(U=0|R,Z,X) \\
        & = exp(\beta_{W0} + \beta_{Wx}X + \beta_{Wu}) P(U=1|R,Z,X) + exp(\beta_{W0} + \beta_{Wx}X)(1-P(U=1|R,Z,X)) \\
        & = exp(\beta_{W0} + \beta_{Wx}X) \bigg[ (exp(\beta_{Wu}) - 1) P(U=1|R,Z,X) + 1\bigg]
    \end{align*}}

    This implies that
    \begin{align*}
        P(U=1|R,Z,X) = \bigg\{\frac{P(W=1|R,Z,X)}{exp(\beta_{W0} + \beta_{Wx}X)} - 1 \bigg\}\frac{1}{(exp(\beta_{Wu}) - 1)}
    \end{align*}

    Then, we have
    \begin{align*}
        E[exp(\beta_{T_u} U) | R,Z,X] & = (exp(\beta_{T_u}) - 1)P(U=1|R,Z,X) + 1 \\
        & = \bigg\{\frac{P(W=1|R,Z,X)}{exp(\beta_{W0} + \beta_{Wx}X)} - 1 \bigg\}\frac{exp(\beta_{T_u}) - 1}{exp(\beta_{Wu}) - 1} + 1 \\
        & = \frac{P(W=1|R,Z,X)}{exp(\beta_{W0} + \beta_{Wx}X)}\frac{exp(\beta_{T_u}) - 1}{exp(\beta_{Wu}) - 1} - \frac{exp(\beta_{T_u}) - 1}{exp(\beta_{Wu}) - 1}  + 1 \\
        & = \frac{P(W=1|R,Z,X)}{exp(\beta_{W0} + \beta_{Wx}X)}\frac{exp(\beta_{T_u}) - 1}{exp(\beta_{Wu}) - 1} + \frac{exp(\beta_{Wu}) - exp(\beta_{T_u})}{exp(\beta_{Wu}) - 1}
    \end{align*}

    Then,
    \begin{align*}
        f(t|R,Z,X) & \approx \lambda_0(t) exp(\beta_{Tx}X) E[exp(\beta_{Tu}U)|R,Z,X] \\
        & = \lambda_0(t) exp(\beta_{Tx}X) (\frac{P(W=1|R,Z,X)}{exp(\beta_{W0} + \beta_{Wx}X)}\frac{exp(\beta_{T_u}) - 1}{exp(\beta_{Wu}) - 1} + \frac{exp(\beta_{Wu}) - exp(\beta_{T_u})}{exp(\beta_{Wu}) - 1}) \\
        & = \lambda_0(t) exp(\beta_{Tx}X)(\alpha \frac{P(W=1|R,Z,X)}{exp(\beta_{W0} + \beta_{Wx}X)} + 1-\alpha)
    \end{align*}
    where $\alpha = \frac{exp(\beta_{Tu}) - 1}{exp(\beta_{Wu}) - 1}$.

\section{Simulation}\label{supp:sim}

We examine the finite sample performance of the semiparametric IPCW approach and the parametric two-stage approach through a simulation study.
We first simulate two binary covariates with the following distribution $X_1\sim {\rm Binom}(n, p=0.7)$ and $X_2\sim {\rm Binom}(n,p=0.5)$. We simulate a continuous $U$ where $U$ follows a truncated normal distribution with $E(U)=0.5$, $Var(U)=1$ and $U\in[0,1]$. Then, we generate a binary $U_b=I(U>0.5)$. The binary negative control outcome $W$ follows a log-linear model where $P(W=1)=exp(\beta_{W0}+\beta_WU+0.2X_1+0.5X_2)$. We simulate binary $R$ and $Z$ with $P(R=1)=1/(1+exp(-1.2+2U_b+0.3X_1+0.2X_2))$ and $P(Z=1)=1/(1+exp(\beta_{Z0}+\beta_ZU_b+0.5X_1-0.2X_2))$. We simulate the true survival and censoring time, $T^*(0)$, $T^*(1)$, $C(0)$ and $C(1)$, with the (mean survival time, mean censoring time) being (2000, 30) days and (6000, 50) days for data with $R=1$ and $R=0$ respectively. The hazard ratios for the true survival time are simulated by $exp(-1+0.8X_1+0.3X_2-2U)$ and $exp(-1+0.5X_1-0.2X_2-1.8U)$ for data with $R=1$ and $R=0$. Similarly, the hazard ratios for the true censoring time are simulated by $exp(-2.5+0.2Z+0.2X_1-0.1X_2)$ and $exp(-2.5+0.2Z+0.3X_1+0.1X_2)$ for data with $R=1$ and $R=0$ respectively. The observed time is $T=R\cdot \min (T^*(1), C(1))+(1-R)\cdot \min (T^*(0), C(0))$ with the event indicator $\Delta=I(T^*\leq C)$. Our estimand is the counterfactual 1-year HIV-free probability $P(T^*(0)>365\mid R=0)$ (equivalently, one minus the incidence rate). Note that this data-generating process satisfies all assumptions in Sections \ref{sec:semiparametric}  and \ref{sec:twostage}.  Our simulation study can then be summarized by the factorial design below. 

Factor 1 (sample size): (1) $n=6500$; (2) $n=10000$; 

Factor 2 (Strengths of negative control variables $W,Z$ as a proxy to unmeasured confounder; \{($\beta_{W0}$, $\beta_W$), ($\beta_{Z0}$, $\beta_Z$)\}: (1) \{medium, medium\}: $\{(-4,3), (-1.2,1)\}$; (2) \{medium, high\}:  $\{(-4,3), (-1.2,2)\}$; (3) \{high, medium\}:  $\{(-4.7,4), (-1.2,1)\}$; (4) \{high, high\}:  $\{(-4.7,4), (-1.2,2)\}$. 

Factor 3 (Estimation methods of counterfactual HIV incidence under placebo arm): (1) oracle estimator adjusting for $(X, U)$; (2) naive method adjusting for $X$; (3) naive method adjusting for $(X, Z, W)$;  (4) semiparametric IPCW using outcome bridge, propensity bridge, and doubly robust estimators; (5) regression-based two-stage approaches. 

Across all data-generating processes, we have $P(R=1)\approx 0.33$, $P(W=1)\approx 0.15$, $P(Z=1)\approx 0.50$, which are similar to our real data example. The true counterfactual HIV cumulative incidence through 1 year is $0.035$ across all settings. Consistent with the real-data analysis, estimations and confidence intervals are obtained using the \texttt{gmm} package in \textsf{R} for all methods. We also examined our statistical inference method with power and type I error rate as described in Section \ref{sec:inference}.

The results are in Table \ref{tab:n6500} and Table \ref{tab:n10000} with $3000$ simulation runs. We summarize (i) the mean of the estimated counterfactual HIV cumulative incidence through 1 year on the $\log(-\log(\cdot))$ scale; (ii) the transformed counterfactual HIV cumulative incidence through 1 year; (iii) the standard deviation of the estimated counterfactual HIV cumulative incidence through 1 year on the $\log(-\log(\cdot))$ scale; (iv) the median of the standard error; (v) empirical coverage of the 95\% confidence interval; (vi) empirical power when estimating $P(T^*(1)>365\mid R=0)-P(T^*(0)>365\mid R=0)$ under the alternative; (vii) empirical type-I error rate when estimating when estimating $P(T^*(1)>365\mid R=0)-P(T^*(0)>365\mid R=0)$ under the null. The IPCW estimates might exceed the $[0,100]$ plausible range of the estimated cumulative incidence, which is considered as a limitation of the IPCW approach in Table \ref{table:comparison}. The results reported for the IPCW approach are obtained across the simulation runs that are within the $[0,100]$ range, and we report the proportion of the excluded invalid estimates in the last column.

\begin{table}[!htbp]
\caption{Simulation result for Factor 1: $n=6500$.}\label{tab:n6500}
\resizebox{\textwidth}{!}{
\begin{tabular}{ccccccccccc}
\toprule
\textbf{Factor 2 }           & \textbf{Method} & \textbf{Specification} & \textbf{Estimate} & \textbf{Est. Incidence} & \textbf{SD} & \textbf{Med. SE} & \textbf{Coverage} & \textbf{Power} & \textbf{Type I Error} & \textbf{Prop\_ex} \\
\midrule
\multirow{7}{*}{medium W, medium Z}   & Oracle & Adjust for X,U   & 1.212   & 0.035  & 0.034       & 0.034     & 0.944  & 1  & 0.061  & -   \\ & Naïve           & Adjust for X    & 1.084   & 0.052 & 0.025       & 0.025     & 0.002     & 1              & 0.938       & -     \\    & Naïve           & Adjust for X,Z,W       & 1.101             & 0.049      & 0.026       & 0.026              & 0.019    & 1   & 0.894     & -                  \\      & IPCW    & Outcome bridge      & 1.207         & 0.035    & 0.208       & 0.205       & 0.954       & 0.515          & 0.047         & 0.119   \\        & IPCW            & Propensity bridge      & 1.206             & 0.035        & 0.206       & 0.205              & 0.953             & 0.518          & 0.048      & 0.119      \\          & IPCW            & Doubly robust          & 1.207             & 0.035         & 0.208       & 0.205              & 0.954             & 0.516          & 0.047         & 0.119              \\   & Two-stage       & -      & 1.204             & 0.036        & 0.128       & 0.121         & 0.973             & 0.492          & 0.003                 & -       \\\midrule   \multirow{7}{*}{medium W, high Z}  & Oracle    & Adjust for X,U         & 1.212        & 0.035              & 0.033       & 0.033         & 0.942     & 1              & 0.061         & -      \\       & Naïve     & Adjust for X   & 1.086             & 0.052        & 0.025       & 0.025      & 0.003      & 1       & 0.938      & -      \\      & Naïve     & Adjust for X,Z,W       & 1.112     & 0.048   & 0.027       & 0.027              & 0.054      & 1    & 0.844       & -    \\    & IPCW     & Outcome bridge   & 1.229     & 0.033    & 0.143       & 0.118        & 0.958     & 0.726  & 0.043       & 0.013     \\      & IPCW     & Propensity bridge      & 1.229       & 0.033      & 0.142       & 0.117     & 0.958  & 0.726          & 0.045     & 0.012     \\     & IPCW     & Doubly robust    & 1.229   & 0.033      & 0.143       & 0.118   & 0.958    & 0.726    & 0.045     & 0.013     \\    & Two-stage       & -      & 1.187     & 0.038    & 0.061       & 0.059   & 0.940             & 0.887          & 0.013    & -     \\ \midrule \multirow{7}{*}{high W, medium Z}  & Oracle          & Adjust for X,U         & 1.212 & 0.035         & 0.034       & 0.034      & 0.944             & 1    & 0.061    & -    \\     & Naïve  & Adjust for X    & 1.084   & 0.052         & 0.025       & 0.025     & 0.002      & 1      & 0.938   & -    \\    & Naïve   & Adjust for X,Z,W       & 1.109    & 0.048      & 0.026       & 0.026       & 0.036  & 1    & 0.842    & -       \\    & IPCW  & Outcome bridge      & 1.214      & 0.034         & 0.213       & 0.204        & 0.952   & 0.509          & 0.048      & 0.106     \\      & IPCW    & Propensity bridge      & 1.216      & 0.034    & 0.215       & 0.203    & 0.952       & 0.508          & 0.050      & 0.104         \\     & IPCW   & Doubly robust   & 1.214    & 0.034      & 0.213       & 0.204      & 0.952    & 0.508   & 0.051   & 0.106     \\   & Two-stage   & -     & 1.185    & 0.038   & 0.060       & 0.122    & 0.969   & 0.473      & 0.002    & -     \\ \midrule  \multirow{7}{*}{high W, high Z} & Oracle      & Adjust for X,U    & 1.212    & 0.035   & 0.033       & 0.033     & 0.942     & 1     & 0.061   & -     \\    & Naïve      & Adjust for X   & 1.086      & 0.052    & 0.025       & 0.025     & 0.003      & 1      & 0.938     & -    \\     & Naïve     & Adjust for X,Z,W   & 1.119      & 0.047  & 0.027       & 0.027       & 0.087    & 1      & 0.789     & -    \\      & IPCW   & Outcome bridge    & 1.227     & 0.033     & 0.136       & 0.117       & 0.962    & 0.742  & 0.037     & 0.010    \\    & IPCW      & Propensity bridge      & 1.228    & 0.033    & 0.137       & 0.117      & 0.962    & 0.741   & 0.041    & 0.010      \\  & IPCW      & Doubly robust    & 1.227   & 0.033       & 0.136       & 0.117      & 0.962      & 0.741    & 0.041    & 0.010    \\      & Two-stage       & -   & 1.196    & 0.037        & 0.101       & 0.060    & 0.942      & 0.870     & 0.012     & -  \\  \bottomrule
\end{tabular}}
\begin{tablenotes}
\small
\item IPCW = inverse probability of censoring weighted; $X$ = baseline covariates; $Z$ = NCE ; $W$ = NCO.
\end{tablenotes}
\end{table}

\begin{table}[!htbp]
\caption{Simulation result for Factor 1: $n=10000$.}\label{tab:n10000}
\resizebox{\textwidth}{!}{
\begin{tabular}{ccccccccccc}
\toprule
\textbf{Factor 2} & \textbf{Method} & \textbf{Specification} & \textbf{Estimate} & 
\textbf{Est. Incidence} & \textbf{SD} & \textbf{Med. SE} & 
\textbf{Coverage} & \textbf{Power} & \textbf{Type I Error} & \textbf{Prop\_ex} \\
\midrule

\multirow{7}{*}{medium W, medium Z}  
& Oracle  & Adjust for X,U     & 1.212 & 0.035 & 0.027 & 0.027 & 0.948 & 1     & 0.040 & - \\
& Naïve   & Adjust for X       & 1.084 & 0.052 & 0.020 & 0.020 & 0.000 & 1     & 0.995 & - \\
& Naïve   & Adjust for X,Z,W   & 1.101 & 0.049 & 0.021 & 0.021 & 0.002 & 1.000 & 0.981 & - \\
& IPCW    & Outcome bridge     & 1.221 & 0.034 & 0.182 & 0.167 & 0.949 & 0.579 & 0.052 & 0.068 \\
& IPCW    & Propensity bridge  & 1.221 & 0.034 & 0.193 & 0.166 & 0.949 & 0.580 & 0.054 & 0.067 \\
& IPCW    & Doubly robust      & 1.221 & 0.034 & 0.182 & 0.167 & 0.948 & 0.578 & 0.054 & 0.068 \\
& Two-stage & -                & 1.185 & 0.038 & 0.050 & 0.096 & 0.963 & 0.625 & 0.003 & - \\
\midrule

\multirow{7}{*}{medium W, high Z}  
& Oracle  & Adjust for X,U     & 1.212 & 0.035 & 0.027 & 0.027 & 0.951 & 1     & 0.048 & - \\
& Naïve   & Adjust for X       & 1.084 & 0.052 & 0.020 & 0.020 & 0.000 & 1     & 0.998 & - \\
& Naïve   & Adjust for X,Z,W   & 1.113 & 0.048 & 0.021 & 0.022 & 0.007 & 1     & 0.956 & - \\
& IPCW    & Outcome bridge     & 1.221 & 0.034 & 0.107 & 0.095 & 0.958 & 0.847 & 0.037 & 0.002 \\
& IPCW    & Propensity bridge  & 1.220 & 0.034 & 0.107 & 0.095 & 0.958 & 0.848 & 0.040 & 0.002 \\
& IPCW    & Doubly robust      & 1.221 & 0.034 & 0.107 & 0.095 & 0.958 & 0.847 & 0.040 & 0.002 \\
& Two-stage & -                & 1.199 & 0.036 & 0.099 & 0.048 & 0.936 & 0.938 & 0.017 & - \\
\midrule

\multirow{7}{*}{high W, medium Z}  
& Oracle  & Adjust for X,U     & 1.212 & 0.035 & 0.027 & 0.027 & 0.948 & 1     & 0.040 & - \\
& Naïve   & Adjust for X       & 1.084 & 0.052 & 0.020 & 0.020 & 0.000 & 1     & 0.995 & - \\
& Naïve   & Adjust for X,Z,W   & 1.109 & 0.048 & 0.021 & 0.021 & 0.003 & 1.000 & 0.960 & - \\
& IPCW    & Outcome bridge     & 1.227 & 0.033 & 0.187 & 0.168 & 0.951 & 0.570 & 0.045 & 0.056 \\
& IPCW    & Propensity bridge  & 1.227 & 0.033 & 0.187 & 0.168 & 0.951 & 0.571 & 0.051 & 0.056 \\
& IPCW    & Doubly robust      & 1.227 & 0.033 & 0.187 & 0.168 & 0.951 & 0.570 & 0.051 & 0.056 \\
& Two-stage & -                & 1.199 & 0.036 & 0.099 & 0.098 & 0.961 & 0.604 & 0.004 & - \\
\midrule

\multirow{7}{*}{high W, high Z}  
& Oracle  & Adjust for X,U     & 1.212 & 0.035 & 0.027 & 0.027 & 0.951 & 1     & 0.048 & - \\
& Naïve   & Adjust for X       & 1.084 & 0.052 & 0.020 & 0.020 & 0.000 & 1     & 0.998 & - \\
& Naïve   & Adjust for X,Z,W   & 1.119 & 0.047 & 0.021 & 0.022 & 0.020 & 1     & 0.935 & - \\
& IPCW    & Outcome bridge     & 1.220 & 0.034 & 0.104 & 0.094 & 0.962 & 0.852 & 0.036 & 0.002 \\
& IPCW    & Propensity bridge  & 1.220 & 0.034 & 0.105 & 0.094 & 0.962 & 0.854 & 0.040 & 0.001 \\
& IPCW    & Doubly robust      & 1.220 & 0.034 & 0.104 & 0.094 & 0.962 & 0.852 & 0.040 & 0.002 \\
& Two-stage & -                & 1.184 & 0.038 & 0.049 & 0.047 & 0.938 & 0.913 & 0.019 & - \\
\bottomrule

\end{tabular}}
\begin{tablenotes}
\small
\item IPCW = inverse probability of censoring weighted; $X$ = baseline covariates; $Z$ = NCE ; $W$ = NCO.
\end{tablenotes}
\end{table}

Observations:
\begin{enumerate}
    \item All IPCW and two-stage estimates are unbiased, while estimates from naïve regressions are biased. Additionally adjusting for $Z,W$ in naive regressions reduce the bias.
    \item Regression-based two-stage estimates have smaller standard deviations and standard errors on average compared to IPCW estimates. 
    \item We observe smaller standard deviations and higher power for both IPCW and two-stage regression methods as sample size increases.
    \item The empirical coverages are approximately at the nominal 95\% and the type I error is controlled at 0.05 for both IPCW and two-stage regression across all settings.
    \item Using negative control variables $W,Z$ that are stronger proxies of the unmeasured confounder returns smaller standard deviations on average for all IPCW and two-stage regression estimates.
\end{enumerate}

\end{document}